\documentclass[12pt]{article}
\usepackage{mathpazo}
\usepackage[latin9]{inputenc}
\usepackage{xcolor}
\usepackage{verbatim}
\usepackage{amsmath}
\usepackage{amsthm}
\usepackage{amssymb}
\usepackage{graphicx}
\usepackage{setspace}
\usepackage{wasysym}
\makeatletter
\theoremstyle{definition}
\newtheorem{defn}{\protect\definitionname}
\theoremstyle{plain}
\newtheorem{thm}{\protect\theoremname}
\theoremstyle{remark}
\newtheorem{rem}{\protect\remarkname}
\theoremstyle{definition}
 \newtheorem{example}{\protect\examplename}
\theoremstyle{plain}
\newtheorem{cor}{\protect\corollaryname}

\usepackage{calc}
\usepackage{amsfonts}
\usepackage{xcolor}
\providecommand{\U}[1]{\protect\rule{.1in}{.1in}}

\providecommand{\corollaryname}{Corollary}
\providecommand{\definitionname}{Definition}
\providecommand{\examplename}{Example}

\providecommand{\remarkname}{Remark}
\providecommand{\theoremname}{Theorem}

\DeclareMathOperator*{\argmax}{argmax}

\ifdefined\showcaptionsetup
 
\fi

\makeatother

\providecommand{\corollaryname}{Corollary}
\providecommand{\definitionname}{Definition}
\providecommand{\examplename}{Example}
\providecommand{\remarkname}{Remark}
\providecommand{\theoremname}{Theorem}

\begin{document}
\title{\vspace{-2cm}
 Identifying Behavioral Types}
\author{Christopher Kops, Paola Manzini, Marco Mariotti, Illia Pasichnichenko\thanks{Kops: Maastricht University, the Netherlands (e-mail: j.kops@maastrichtuniversity.nl);
Manzini: University of Bristol, United Kingdom (email: p.manzini@bristol.ac.uk)
and IZA Bonn; Mariotti: Queen Mary University of London, United Kingdom
(email: m.mariotti@qmul.ac.uk) and Deakin University; Pasichnichenko:
University of Sussex, United Kingdom (email: i.pasichnichenko@sussex.ac.uk).
We thank audiences at Sussex, Maastricht, Karlsruhe, the 11th BRIC
Conference at ITAM, and the 2025 Unforeseen Contingencies Workshop
at UGA.}}

\maketitle
\vspace{-0.5cm}

\begin{abstract}
We study identification in models of aggregate choice generated by
unobserved behavioral types. An analyst observes only aggregate choice
behavior, while the population distribution of types and their type-level
choice patterns are latent. Assuming only minimal and purely qualitative
prior knowledge of the process generating type-level choice probabilities,
we characterize necessary and sufficient conditions for identifiability.
Identification obtains if and only if the data exhibit sufficient
cross-type behavioral heterogeneity, which we characterize equivalently
through combinatorial matching conditions between types and alternatives,
and through algebraic properties of the matrices mapping type-level
to aggregate choice behavior.\\
 \textbf{JEL codes:} D81, D83, D91\\
 \textbf{Keywords:} revealed preferences, stochastic choice, bounded
rationality, attention, characteristics, incomplete preferences.
\end{abstract}

\section{Introduction}

\label{sec:intro} Latent factors such as preferences, attention,
and cognitive capacity influence choice behavior. These factors are
typically heterogeneously distributed---be it across individuals
in a population or within a single decision maker observed repeatedly
over time. This heterogeneity then gives rise to various \emph{behavioral
types}. Whenever the individual choices of these types remain unobservable---for
example because of privacy protection, anonymization, or the sheer
scale of data---they may still produce observable choice data on
the aggregate level. This raises a fundamental question: when is it
possible to infer the distribution of behavioral types, or their type-specific
choice patterns, from aggregate choice data alone? This paper characterizes
conditions under which such inferences can be drawn, assuming minimal
and purely qualitative prior knowledge of the process generating type-level
choice probabilities.

To illustrate the motivation, consider a policymaker contemplating
a stimulus that expands the set of options available to consumers.
Assessing the effectiveness of such a policy requires knowledge of
the distribution of consumer types. Some types may attend to the full
range of available options, others may focus only on a subset, and
still others may consistently adhere to their current choice. The
overall impact of the policy depends both on the prevalence of these
behavioral types and on the manner in which each type responds to
the expanded choice set. %
{} %

We study a general framework of stochastic choice. The analyst observes
choice probabilities (interpreted as idealized empirical \emph{shares})
over a finite set of $n$ alternatives. These shares result from the
aggregation of the unobserved stochastic choices made by a finite
number of behavioral types. Formally, the core framework is captured
by the equation
\begin{equation}
p=M\pi\label{eq:CoreModel}
\end{equation}
where $p$ is an observed $n\times1$ vector of choice probabilities
over the alternatives listed in a given order, and $M$ is an $n\times r$
matrix describing the partially unknown choice probabilities of $r$
\emph{types}. A type $t$ is a probability distribution over the alternatives.
The entry $M\left(k,t\right)$ of the matrix $M$ denotes the probability
with which type $t$ chooses alternative $k$. A theory of choice
specifies how types choose, that is, the entries in $M$. Finally,
$\pi$ is an unobserved $r\times1$ vector of probabilities expressing
the distribution of types. Since often stochastic choice data are
observed in multiple occasions of choice, we also consider an extension
of the model in Equation \eqref{eq:CoreModel} in this direction.

A key feature of our approach is that we impose only minimal informational
requirements on the analyst\textquoteright s knowledge of the choice
process. Specifically, the analyst is assumed to have prior knowledge
only of the \emph{possibility relation} associated with type-level
choice behavior---namely, which types may choose a given alternative
with positive probability and which types cannot possibly do so. Such
information is often implicit in the definition of behavioral types
in given contexts: vegan consumers do not choose animal products;
risk-averters steer clear of junk bonds; information avoiders omit
health screenings or ignore adverse media sources. At a theoretical
level, behavioral models may similarly restrict the support of type-level
randomization, for instance by allowing randomization only within
a consideration set or among maximal elements of an incomplete preference
relation.\footnote{Concrete examples are studied in more detail in Section \ref{sec:Applications}.}

Our focus is on how the data in $p$ can be used to identify the unknown
parameters---namely, the elements of $\pi$ or the positive elements
of $M$.\footnote{In the next section we clarify the notion of genericity used in the
paper.} Under the informational assumption discussed above, we derive identification
results in three settings. First, we characterize identification of
the type distribution $\pi$ without imposing further restrictions
on the structure of type-level choice behavior. 

Second, we specialize the framework to a type--state model, in which
the stochastic component of each type\textquoteright s behavior depends
on the realization of a common random state that is unobserved by
the analyst. In this setting, aggregate choice probabilities arise
as weighted sums of state-dependent choice matrices, according to
the equation
\[
p=\left(\sum_{a\in A}f(a)M^{a}\right)\pi
\]
where $A$ is a set of states, $f\left(a\right)$ is the probability
of state $a$, and $M^{a}$ is a matrix of $0-1$ type-level choice
probabilities that applies in state $a$. In this model, we obtain
a further set of conditions for the identification of $\pi$.

Finally, in Section \ref{sec:Inference-from-multiple}, we demonstrate
that observing choices across multiple occasions strengthens identification,
as it allows the analyst to exploit cross-temporal variation. For
this purpose, we combine the new one-shot identification conditions
with a fundamental result from tensor decompositions. This leads to
sufficient conditions for the identification of both the type distribution
and the type-level choice probabilities.

Our identification results can be broadly understood as formalizing
a requirement of sufficient \emph{behavioral heterogeneity}. They
fall into two main classes. The first class provides a combinatorial
characterization, based on the existence of suitable \emph{matchings}
between types and alternatives that are consistent with the possibility
relation. These matchings capture heterogeneity in the sense that
different types can, in principle, be associated with distinct choices
when there is sufficient behavioral richness. The second class provides
a geometric and algebraic counterpart to this analysis. We establish
new results on the generic invertibility of weighted sums of state-dependent
choice matrices, formulated as restrictions on the nullspaces of the
matrices $M^{a}$. These algebraic conditions clarify \emph{why} identification
failures arise---by making explicit the directions in aggregate choice
space along which cancellations occur---and admit a natural interpretation
as requirements of behavioral heterogeneity. 

The simplest case---consisting of a single choice occasion, two states
$a$ and $b$, two types, and two alternatives---already illustrates
the multifaceted nature of our results. In this setting, generic (i.e.,
holding for all parameter values except possibly on a measure-zero
set) identifiability of the type distribution is equivalent to either
of the following conditions: (i) the nullspaces of $M^{a}$ and $M^{b}$
intersect only trivially; (ii) in at least one of the two states,
there exists a matching between types and alternatives. The behavioral
meaning of both conditions is that at least one state must ``separate''
the two types, in the sense that the types choose differently in that
state. If, in addition, each type's behavior is constant across states,
then identification becomes global---i.e., it holds for all parameter
values rather than merely generically.

Overall, our results provide an exhaustive characterization of the
scope for identification. It is important to emphasize that the notion
of behavioral heterogeneity across types at the root of this characterization
is purely \emph{qualitative}: our identifying conditions do not depend
quantitatively on the type-level choice probabilities, but only on
the pattern of structural zeros in choices. Equivalently, the entire
analysis depends solely on the sign pattern of the matrix $M$.

The scope and applicability of our results is demonstrated by two
extended examples in Section \ref{sec:Applications}. The first examines
choice between vectors of characteristics \emph{à la} Lancaster \cite{Lancaster1966},
while the other studies a model of incomplete preferences.

\subsection{Related literature}

Dardanoni \emph{et al.} \cite{DardanoniManziniMariottiTyson2020}
has a similar motivation to ours, in that it studies generic type
identification from discrete choice data generated by a single menu
without additional covariate information, linking identifiability
to the invertibility of an associated matrix. However their analysis
is concerned with a specific parametric model of choice with consideration
sets, in which types are indexed by consideration capacities---that
is, by the maximum number of alternatives a decision maker can consider.
In that setting, the matrix representation of the model has a known
and highly structured form. A more econometrics-facing work in the
same vein, which however makes use of covariates besides choice observations
from a single menu with limited consideration (or limited feasibility)
is Barseghyan \emph{et. al.} \cite{barcoumoltei21}. By contrast,
the framework developed here is deliberately abstract and qualitative.
Our objective is not to analyze a particular behavioral model, but
to accommodate a broad class of models as special cases and to characterize
identification at this level of generality and exploiting only the
qualitative nature of the process of choice. 

A recent strand of the literature studies the identifiability of the
preference distribution from choice data within Random Utility Maximization
(RUM) frameworks. Contributions in this vein include Apesteguia \emph{et
al.} \cite{apeballu17}, Chambers and Turansick \cite{Chambers25},
Doignon and Saito \cite{doisai23}, Filiz-Ozbay and Masatlioglu \cite{filmas23},
Lu \cite{Lu2019}, Manzini and Mariotti \cite{manmar18},\footnote{See also the Correction in Manzini \emph{et al.} \cite{manmarpet19}.}
Turansick \cite{tur22}, and Yildiz \cite{yil22} (see also the important
earlier piece by Fishburn \cite{fish98}). Our model shares with this
literature a broad objective within a stochastic choice framework,
namely the identification of underlying heterogeneity from observed
choices. Our analysis departs from this literature in two main respects.
First, we do not take preferences, rankings, or general choice functions
as primitives; instead, we work with abstract \textquotedblleft behavioral
types.\textquotedblright{} Second, standard RUM approaches typically
take as observable a random choice rule specifying a distribution
of choices from each of several overlapping menus---distinct subsets
of a common grand set. By contrast, the core of our analysis relies
on the observed choice distribution from a single set of alternatives,
although we also provide an extension to a parsimonious multi-occasion
framework. These differences necessitate distinct analytical tools
and arguments.

Caradonna and Turansick \cite{cartur24}, like the papers above, also
focus on RUM identification, but this is mainly for expositional purposes.
Their work develops new general techniques for identification (related
to the unique decomposition of flows on directed acyclic graphs into
measures over paths) that extend to more general random choice models,
as well as dynamic discrete choice models. Once again, a main difference
from our approach lies in their use of a random choice rule as a primitive.

Progress on identification with more limited data than a full random
choice rule can be credited to Azrieli and Rehbeck \cite{azrreh22}.
Their framework is one of \emph{marginal stochastic choice}, a setting
in which only the probability of choosing each alternative and the
probability of each menu\textquoteright s realization are observed.

Our multi-occasion extension of the main framework shares methodological
affinities with Allman \emph{et al.} \cite{AllmanMatiasRhodes2009}
and Dardanoni \emph{et al.} \cite{darmanmarpettys24}, which in turn
build on classical tensor decomposition results by Kruskal \cite{Kruskal1977},
with a recent proof by Rhodes \cite{Rhodes2010}.

At a much broader level our work also relates to latent structure
models, running from finite mixture models (McLachlan and Peel \cite{McLachlan})
to nonparametric mixture models (Lindsay \cite{Lindsay1995}), and
hidden Markov models (Cappé \emph{et al.} \cite{Cappe2005}, Ephraim
and Merhav \cite{Ephraim2002}). The identification problem of mapping
observed choice data back to latent behavioral types finds its formal
roots in the seminal contributions of Koopmans and Reisersøl \cite{Koopmans1950a}
and Koopmans \cite{Koopmans1950b}. Traditionally, the study of identifiability
assumes a hypothetical, exact knowledge of the distribution of observables.
The central question then is whether the underlying structural parameters
can be uniquely recovered from this distribution. While such identification
is distinct from statistical inference, it remains a prerequisite
for it, as non-identifiable parameters render consistent estimation
impossible.

While we focus entirely a on a qualitative notion of behavioral heterogeneity,
a complementary approach is taken by Apesteguia and Ballester \cite{Apesteguia2025},
who study the \emph{measurement} of behavioral heterogeneity in environments
closely related to those underlying our identification analysis.

\section{Framework and definitions}

\label{sec:Framework-and-definitions} 

Let $T=\left\{ t_{1},\dots,t_{r}\right\} $ and $X=\left\{ x_{1},\dots,x_{n}\right\} $
be finite sets of types and alternatives, respectively. The analyst
observes choice shares $p$ for all alternatives in $X$. Let $M$
be an $n\times r$ matrix with a typical column containing the choice
shares for a specific type, so that the data can be described by Equation
\eqref{eq:CoreModel}. 

As noted, similarly to the RUM framework, this model admits both a
population and an intrapersonal interpretation. Under the population
interpretation, type heterogeneity captures variation in preferences
or cognitive traits across individuals. Under the intrapersonal interpretation,
the types instead represent distinct choice-relevant states---preferential
or cognitive---in which a single decision maker may find themselves
at different points in time.
\begin{defn}
\label{def:possible type}A type $t_{l}$ is \emph{possible} \emph{at
alternative $x_{k}$} if $M\left(k,l\right)>0$. In this case we write
$t_{l}\land x_{k}$, where $\land$ is a binary relation over the
sets $X$ and $T$.
\end{defn}
The possibility relation $\wedge$ captures the analyst\textquoteright s
prior knowledge about the structural zeros of the type--conditional
choice matrix $M$, while leaving unrestricted the magnitudes of the
positive entries. Thus, we identify the parameter space of the model
with a full dimensional subset of $\left[0,1\right]^{\left|\land\right|-1}$,
where $r-1$ dimensions describe $\pi$ and $\left|\land\right|-r$
dimensions describe nonzero entries of $M$. The model's parameterization
map is defined on the parameter space and takes values in $\left[0,1\right]^{n}$.
The model is \emph{globally identifiable} if its parameterization
map is injective. Often we will refer to identifiability as a property
of (a subset of) the parameters in the model.

In this paper we focus mostly on \emph{generic identifiability}. This
implies that the set of parameter values for which identifiability
does not hold has measure zero. In this sense, generic identifiability
is often sufficient for data analysis purposes. We next briefly lay
out the notions from algebraic geometry used repeatedly throughout
the paper. For more thorough introductions to the field, see Allman
\emph{et al.} \cite{AllmanMatiasRhodes2009} and Cox \emph{et al.}
\cite{Cox97}.

An algebraic variety is defined as the simultaneous zero-set of a
finite collection of multivariate polynomials $\{\phi_{i}\}$ on $\mathbb{R}^{m}$.
A variety is all of $\mathbb{R}^{m}$ only when all $\phi_{i}$ are
0; otherwise, a variety is called a proper subvariety and must be
of dimension less than $m$, and, hence, of Lebesgue measure 0 in
$\mathbb{R}^{m}$. The same is true if we replace $\mathbb{R}^{m}$
by a full-dimensional subset of $[0,1]^{m}$. Intersections of algebraic
varieties are algebraic varieties as they are the simultaneous zero-set
of the unions of the original sets of polynomials. Finite unions of
varieties are also varieties, since if sets $G_{1}$ and $G_{2}$
define varieties, then $\{fg|f\in G_{1},g\in G_{2}\}$ defines their
union. Given a set $\Theta\subseteq\mathbb{R}^{m}$ of full dimension,
we say that a property holds \emph{generically} on $\Theta$ if it
holds for all points in $\Theta$, except possibly for those on some
proper subvariety. Thus, generic identifiability of parameters means
that all non-identifiable parameter choices lie within a proper subvariety,
which is a set of Lebesgue measure zero.

Finally, the model with parametrization map $z$ is \emph{structurally
non-identifiable} if no point in $\Theta$ has a singleton fiber,
that is, for all $\theta\in\Theta$, $\left|{\cal F}\left(\theta\right)\right|>1$,
where ${\cal F}\left(\theta\right)=\left\{ \theta'\in\Theta|z\left(\theta'\right)=z\left(\theta\right)\right\} $
is the fiber of $z$ at $\theta$.

\section{Characterizing identifiability of the type distribution}

In this sparse observational context, our primary aim is the identification
of $\pi$. We will show that generic identification is equivalent
to the types' choices being sufficiently heterogeneous in a qualitative
sense. Specifically, we establish a connection between the identifiability
of $\pi$ and the possibility relation $\land.$ 

The following notion describes the possibility for types to choose
distinct alternatives. 
\begin{defn}
\label{def:matching}A (type-to-alternative) \emph{matching} is an
injective function $m:T\to X$ such that $t\land m(t)$ for all $t\in T$.

In our setting, identification of $\pi$ in the model defined by Equation
\eqref{eq:CoreModel} reduces to whether the structural zero pattern
encoded by $\wedge$ admits parameter configurations for which the
columns of $M$ are linearly independent. The first result shows that
this condition admits a simple combinatorial interpretation in terms
of matchings between types and alternatives:
\end{defn}
\begin{thm}
\label{prop:generic-ident}The parameters in $\pi$ are generically
identifiable if and only if there exists a matching $m:T\to X$. In
addition, if a matching exists and it is unique, then the parameters
in $\pi$ are globally identifiable. If no matching exists, then the
parameters in $\pi$ are structurally non-identifiable.
\end{thm}
\begin{proof}
Please refer to Appendix \ref{app:1}.
\end{proof}
Note that the existence of a matching implies $n\geq r$, i.e., there
are at least as many alternatives as types. The interpretation of
the matching is that $r$ agents of distinct types can always choose
$r$ distinct alternatives. In this case, the possibility relation
$\land$ is sufficiently rich to permit full identification of the
type distribution from observed choices. Theorem \ref{prop:generic-ident}
shows that the existence of a matching restricts the admissible behavioral
theories exactly to those in which identification arises almost always.
The theorem also clarifies that the failure of identification on measure
zero sets is attributable to the multiplicity of matchings: when the
matching is unique, identification occurs for all values of the parameters.
As we will show below, the converse is not true.

The necessity of a matching can be understood as follows. If no matching
exists between types and alternatives, then there is a subset of types
whose ``collectively choosable'' alternatives---that is, all those
at which one or more types in the subset are possible---are too few
to accommodate them all. As a result, any attempt to assign distinct
alternatives to these types must necessarily assign at least two types
to the same alternative, or assign some type to an alternative it
can never choose. Translating this into matrix terms, the columns
of $M$ corresponding to these types must be supported on a set of
rows that is not rich enough to ensure generic linear independence.
This issue is \emph{structural}: it does not depend on the particular
values of the positive entries of $M$, but follows solely from the
zero pattern implied by the possibility relation. No selection of
the positive parameter values in $M$ can, in this case, permit learning
of $\pi$ from the data.
\begin{rem}
Formally, the notion of a matching used here is equivalent to that
of a perfect matching in a bipartite graph whose left and right vertices
correspond to columns and rows of the matrix $M$, respectively, with
edges induced by its nonzero entries. Hence, the existence of a matching
is equivalent to the condition in Hall\textquoteright s marriage theorem
(Hall \cite{Hall1935}). Translated to our context, the condition
requires that for any set $S\subseteq T$ of types there exist at
least $\left|S\right|$ distinct alternatives such that some type
in $S$ is possible at at least one of them. This further clarifies
the sense in which the existence of a matching captures behavioral
heterogeneity.
\end{rem}
The following examples illustrate the result.
\begin{example}
\label{ex:M-matrix-3x3}Suppose that there are three types, labeled
$t_{1}$, $t_{2}$, and $t_{3}$, and three alternatives, labeled
$x$, $y$, and $z$. A behavioral theory determines the possibility
relation $\land=\left\{ \left(t_{1},x\right),\left(t_{2},x\right),\left(t_{2},z\right),\left(t_{3},x\right),\left(t_{3},y\right),\left(t_{3},z\right)\right\} $.
In turn this induces the zero entries of $M$: \[
M = \bordermatrix{~ & t_1 & t_2 & t_3 \cr
               x & \color{red}{\bullet} & \bullet & \bullet \cr
               y & 0 & 0 & \color{red}{\bullet} \cr
               z & 0 & \color{red}{\bullet} & \bullet \cr},
\]where black dots ($\bullet$) and red dots (${\color{red}{\bullet}}$)
denote parameters. This encapsulates the information on which alternatives
could possibly have been chosen by which types. The parametric description
of this matrix takes values in $\left[0,1\right]^{3}$. It is clear
that the columns of $M$ can never be linearly dependent, whatever
the values assigned to the dots. To build a matching, observe that
$t_{1}$ must be matched to $x$, because this is the only alternative
for which this type is possible. Then, $t_{2}$ must be matched to
$z$ and $t_{3}$ to $y$. The constructed matching (in red) is unique.
On the other hand, if type $t_{1}$ were possible at $y$, so that
the entry $M\left(2,1\right)$ is positive, there would exist other
matchings, such as $t_{1}$ matched to $y$, $t_{2}$ to $x$ and
$t_{3}$ to $z$. In this case the columns could be linearly dependent
for some values of the parameters, but the set of such values is of
measure zero. 
\end{example}
\begin{example}
Suppose that each type coincides with a consideration capacity, as
in Dardanoni et al. \cite{DardanoniManziniMariottiTyson2020}, namely
the maximum number of alternatives the agent can consider. Let $T=\left\{ t_{1},t_{2},...,t_{n}\right\} $
with type $t_{k}$ having consideration capacity $k$.\footnote{Here, type $t_{n}$ stands for a composite type including all types
that can consider \emph{at least} $n$ alternatives.} Each type maximizes on the set of considered alternatives a known
and common preference, namely a linear order $\succ$ defined on $X$.
Type $t_{k}$ considers each subset of $X$ of cardinality $k$ with
a probability governed by an unknown full-support probability distribution
on $\left\{ K\subseteq X|\,\left|K\right|=k\right\} $. This model
generalizes the fixed-preference model in \cite{DardanoniManziniMariottiTyson2020},
who assume the distribution on each $\left\{ K\subseteq X|\left|K\right|=k\right\} $
to be uniform.

To build a matching, number the alternatives $x_{1},x_{2},...,x_{n}$
in decreasing order of preference. Observe that type $t_{n}$, who
considers the entire set $X$, chooses $x_{1}$ with probability one
and must therefore be matched to $x_{1}$. Type $t_{n-1}$, who considers
all but one alternative, chooses with positive probability only $x_{1}$
and $x_{2}$, so that, in view of the previous matching, must be matched
with $x_{2}$. Continuing in this way, we build a (unique) matching,
with type $t_{k}$ matched to $x_{n-k+1}$ for each $k$. It follows
from Theorem \ref{prop:generic-ident} that the type distribution
is identified. Proposition 2 in \cite{DardanoniManziniMariottiTyson2020}
is a special case of this result. \footnote{What is more, Proposition 6 in \cite{DardanoniManziniMariottiTyson2020},
showing identifiability for generic preference distributions in the
extension of their model to heterogeneous preferences, is also implied
by Theorem \ref{prop:generic-ident}---note that the uniform distribution
assumption on equal-sized consideration sets makes type-level choice
probabilities structurally fixed for any given preference, so that
the only variable parameters in this model are those that relate to
the preference distribution.}
\end{example}

\subsection{Global identifiability}

To provide a full characterization of global identifiability, we introduce
some additional standard terminology. For any given permutation $\sigma$
on the index set $\left\{ 1,2,...,r\right\} $, an inversion is a
pair $\left(i,j\right)$ such that $i<j$ and $\sigma\left(i\right)>\sigma\left(j\right)$.
A permutation is \emph{even (resp., odd) }if it consists of an even
(resp. odd) number of inversions. Any matching $m:T\to X$ is associated
with a permutation $\sigma_{m}$ defined by $m\left(t_{i}\right)=x_{\sigma_{m}\left(i\right)}$.
The \emph{parity} of a matching $m$ is defined as \emph{even }(resp.,
\emph{odd}) if $\sigma_{m}$ is an even (resp., odd) permutation.\footnote{While the parity of a matching depends on the chosen orderings of
types and alternatives, only relative parity between matchings matters
for our results, so the choice of ordering is immaterial.} Two matchings have the same parity if and only if one can be transformed
into the other by an even number of pairwise swaps of assignments
between types. More precisely, a swap is a permutation that exchanges
two elements and leaves all the others fixed, that is for distinct
indices $i\neq j$ the swap between $i$ and $j$ is the permutation
$\sigma$ such that $\sigma\left(i\right)=j$, $\sigma\left(j\right)=i$
and $\sigma\left(k\right)=k$ for all $k\neq i,j$. 

The following result shows that the transition from generic to global
identifiability can be obtained by restricting the way multiple matchings
can coexist.

{\bfseries{}%
}
\begin{thm}
\label{thm:full-id-SNS}The parameters in $\pi$ are globally identifiable
if and only if there exists a subset $R\subseteq X$ with $|R|=r$
such that:

(i) There exists a matching $m:T\to R$;

(ii) Any other matching $m':T\to R$ can be obtained from $m$ via
an even number of swaps.
\end{thm}
\begin{proof}
Please refer to Appendix \ref{App:full-id-SNS}.
\end{proof}
The condition in Theorem 2 coincides with the classical parity characterization
of sign-nonsingular (SNS) matrices. Our result simply adapts this
characterization to a rectangular and column-stochastic setting, and
interprets it as a condition for global identifiability of the type
distribution.\footnote{SNS matrices are square matrices where every real matrix with the
same sign pattern (positive, negative, or zero) is nonsingular. They
were introduced to Economics by Samuelson \cite{sam47}, and later
studied by others (e.g. Bassett \emph{et al.} \cite{basmayqui68},
Gorman \cite{gor64}, Lancaster \cite{Lancaster62}), not in connection
with identification but rather with comparative statics and the correspondence
principle. For a modern mathematical treatment of SNS matrices see
Brualdi and Shader \cite{brusha95}, Chapter 6.} In this interpretation, a matching is a possible ``explanation''
of the data imputing the choice of an alternative to an exclusive
type. The difficulty with two matchings connected by an odd number
of swaps is that they allow a complete reassignment of alternatives
across types that leaves aggregate choice behavior unchanged. Such
a reassignment creates two distinct but observationally equivalent
explanations of the data, undermining global identifiability.

When there exists a unique matching, condition (ii) in the statement
is vacuously satisfied; therefore, Theorem \ref{thm:full-id-SNS}
implies as a special case the sufficient condition for identifiability
of Theorem \ref{prop:generic-ident}.

\section{The type-state framework}\label{sec:One-shot-choice Type-State}

We now endow the framework with additional structure, which is common
in many economic environments. This allows for a more nuanced identification
analysis than that developed in the previous section. We assume that
the realization of a stochastic state determines the nature of the
matrix $M$, while choice by each type is deterministic conditional
on the state. Denote by $f$ a probability measure defined on a known
finite set $A$ of states. Each type $t$ first observes a realization
of $a$ according to $f$. Then, $t$ chooses according to a ``choice
function'' $c_{t}:A\to X$, where $c_{t}\left(a\right)$ denotes
the choice made by type $t$ in state $a$. In this sense, $f$ is
independent from the type.

For a given state $a\in A$ denote by $M^{a}$ the corresponding $n\times r$
matrix of choice shares for each type. A column $l$ of this matrix
is a one-hot vector such that $M^{a}\left(k,l\right)=1\iff c_{t_{l}}\left(a\right)=x_{k}$.
The matrix $M$ is a convex combination of such matrices, 
\begin{equation}
M=\sum_{a\in A}f\left(a\right)M^{a}.\label{eq:matrix-conv-comb}
\end{equation}
In this way, a parametric description of $M$ is given by $\left\{ f\left(a\right)\right\} _{a\in A}$
taking values in $[0,1]^{\left|A\right|-1}$. Here is a simple example
of this structure:\footnote{More elaborate examples are studied in Section \ref{sec:Applications}.}
\begin{example}
Suppose that each state $a$ corresponds to a linear preference order
$\succ_{a}$ over $X$, and that $f\left(a\right)$ gives the probability
assigned to preference $\succ_{a}$. Types differ in their \emph{attention
depth}. Specifically, alternatives are ordered by salience and type
$t_{k}$ searches only the $k$ most salient alternatives. The columns
of $M^{a}$ describe the type-level choices of agents with a given
attention depth. Index alternatives in decreasing order of salience.
Then the choice function of type $c_{t_{k}}$ is determined according
to $c_{t_{k}}\left(a\right)=\arg\max_{\succ_{a}}\left\{ x_{1},...,x_{k}\right\} $,
for $k=1,...n$. That is, for each state $a$, type $t_{k}$ chooses
the most-preferred alternative among the first $k$ alternatives according
to $\succ_{a}$. The resulting state-aggregated choice probabilities
in $M$ are then given by Equation \eqref{eq:matrix-conv-comb}. 
\end{example}
We establish two classes of characterization results that guarantee
generic identification. The first class provides matching conditions
between types and alternatives, in the spirit of those developed in
the previous section. The second class elucidates the \emph{sources}
of identification failure by focusing on algebraic properties of the
matrix $M$---in particular, conditions involving the nullspaces
of its constituent matrices---that determine its generic invertibility
and their behavioral interpretation.

\subsection{Identifiability of the type distribution via matching}

The following definition formalizes the idea of a state in which the
type is revealed through choice.
\begin{defn}
\label{def:ident-given-a}A state $a^{*}\in A$ \emph{separates}\textit{\emph{
types}}\textsl{ }if, for any two types $t_{1}$ and $t_{2}$, we have
$c_{t_{1}}\left(a^{*}\right)=c_{t_{2}}\left(^{*}\right)$ only if
$t_{1}=t_{2}$.
\end{defn}
Thus, ``$a^{*}$ separates types'' means that in state $a^{*}$
there is a matching of the types with the chosen alternatives. When
a state does not separate types we say that it \emph{pools} them.
The following result shows that, within the type-state framework,
the existence of just one type separating state is sufficient for
generic identification. Although this result follows from our later
characterization in Theorem \ref{prop:gen-indep-cols-iff}, it is
of independent interest due to the simplicity of its condition.
\begin{thm}
\label{Th:fully-id-given-a}Suppose there exists a state $a^{*}\in A$
that separates types. Then the parameters in $\pi$ are generically
identifiable.\label{prop:ident-given-a} 
\end{thm}
\begin{proof}
Please refer to Appendix \ref{app:2}.
\end{proof}
To illustrate the previous result, consider the following example.
\begin{example}
There are three types, three alternatives, and two states, $A=\left\{ a,b\right\} $,
and
\[
{\textstyle M^{a}=\left[\begin{array}{ccc}
1 & 0 & 0\\
0 & 0 & 1\\
0 & 1 & 0
\end{array}\right],\quad M^{b}=\left[\begin{array}{ccc}
1 & 0 & 1\\
0 & 0 & 0\\
0 & 1 & 0
\end{array}\right],\quad M=\left[\begin{array}{ccc}
1 & 0 & f(b)\\
0 & 0 & f(a)\\
0 & 1 & 0
\end{array}\right],}
\]
where the third matrix is defined by $M=f(a)M^{a}+f(b)M^{b}$. Clearly,
$a$ separates types but $b$ does not. Because $f(b)=1-f(a)$, the
complete parameter description of $M$ is given by $f(a)\in\left[0,1\right]$.
By Theorem \ref{prop:ident-given-a}, the parameters in $\pi$ are
generically identifiable. Indeed, we can see that $\pi$ is not identifiable
only for $f(a)=0$.
\end{example}
We can make further progress by specifying the definition of a matching
for the current case. 
\begin{defn}
A \emph{matching} is an injective function $m:T\rightarrow X$ such
that for all $t\in T$, $m(t)=c_{t}\left(a\right)$ for some $a\in A$. 
\end{defn}
Note that when $A$ consists of a single state $a$, the existence
of a matching is equivalent to the fact that $a$ separates types.
The type-state framework also allows us to record the states in which
each type is matched to its corresponding alternative, motivating
the following definition.
\begin{defn}
A \emph{state-matching} is a pair $(m,\gamma)$ where $m$ is a matching
and $\gamma:T\to A$ is such that $m(t)=c_{t}\left(\gamma\left(t\right)\right)$
for all $t\in T$. 
\end{defn}
In words, a state-matching is defined by (i) an exclusive assignment
of an alternative to each type; and (ii) a specification of the state
in which any given type selects the assigned alternative. Note that
for a single matching $m$ there may be various state-matchings $(m,\gamma)$,
$(m,\gamma')$, etc.

For any state-matching $(m,\gamma)$, let $\Gamma(m,\gamma)=\left[\nu_{1}\ldots\nu_{|A|}\right]$
be the state-usage vector, indicating the numbers of times each state
$a_{i}$ is used to form the matching $m$, i.e., 
\[
\nu_{i}=\sum_{t\in T}\mathbf{1}\left[\gamma(t)=a_{i}\right].
\]
The following result spells out a condition which is both necessary
and sufficient for the generic identification of $\pi$ in the type-state
framework. For any subset of alternatives $R\subseteq X$ with $|R|=r$,
let $M_{R}$ denote the $r\times r$ submatrix of $M$ obtained by
retaining the rows associated with $R$ (and all $r$ columns). Denoting
$\textrm{\text{Im}}\left(m\right)=\left\{ m\left(t_{}\right)|\,\ensuremath{t_{}\in T}\right\} $
the image of $m$, let 
\[
\mathcal{S}_{R}:=\{(m,\gamma):(m,\gamma)\ \text{is a state-matching and }\text{Im}(m)=R\}.
\]

\begin{thm}
\label{prop:gen-indep-cols-iff} The parameters in $\pi$ are generically
identifiable if and only if there exists a subset $R\subseteq X$
with $|R|=r$ such that the following two conditions hold: 

(i) $\mathcal{S}_{R}$ is nonempty. 

(ii) For any partition of $\mathcal{S}_{R}$ into disjoint pairs $\big((m,\gamma),(m',\gamma')\big)$
with $\Gamma(m,\gamma)=\Gamma(m',\gamma')$, there exists at least
one pair for which $m'$ can be obtained from $m$ via an even number
of swaps. 

Moreover, if no such $R$ exists, then $\pi$ is structurally non-identifiable.
\end{thm}
\begin{proof}
Please refer to Appendix \ref{app:3}.
\end{proof}
The first condition in Theorem \ref{prop:gen-indep-cols-iff} simply
asserts the existence of a state-matching $\left(m,\gamma\right)$
for which the image of $m$ is the subset of alternatives $R$. The
second condition is a restriction on how state-matchings in $\mathcal{S}_{R}$
of different parities can coexist within a given state-usage profile.
In particular, it forbids the possibility that state-matchings of
opposite parity be paired in such a way that every pair is identical
in terms of its state usage. Note that this condition weakens the
one in Theorem \ref{thm:full-id-SNS}, which required \emph{all} matchings
to be obtainable from each other via an even number of swaps. In the
proof, the first condition ensures that at least one nonzero multiplicative
term appears in the Leibniz expansion of the determinant, while the
second guarantees that these terms are not fully cancelled at the
summation stage. Theorem \ref{nullspaces_general} below will offer
a behavioral interpretation and an equivalent algebraic condition.

The next two examples show situations where, respectively, the conditions
of Theorem \ref{prop:gen-indep-cols-iff} are or are not satisfied.
\begin{example}
Consider
\[
\alpha\left[\begin{array}{ccc}
1 & 0 & 0\\
0 & 1 & 1\\
0 & 0 & 0
\end{array}\right]+\beta\left[\begin{array}{ccc}
0 & 1 & 0\\
0 & 0 & 0\\
1 & 0 & 1
\end{array}\right]=\left[\begin{array}{ccc}
\alpha & \beta & 0\\
0 & \alpha & \alpha\\
\beta & 0 & \beta
\end{array}\right],
\]
where $X=\left\{ x,y,z\right\} $, $A=\left\{ a,b\right\} $, $\alpha=f(a)$,
and $\beta=f(b)$. The columns of the right-hand side matrix are generically
independent because the determinant is not zero, for example, for
$\alpha=\beta=\frac{1}{2}$. There is a state-matching $(m,\gamma)=\left(xyz,aab\right)$,
with $\Gamma(m,\gamma)=\left[2\ 1\right]$ . There is another state-matching
$(m',\gamma')=\left(zxy,bba\right)$, with $\Gamma(m',\gamma')=\left[1\ 2\right]$.
Clearly, there is no other state-matchings. Since $\Gamma(m,\gamma)\neq\Gamma(m',\gamma')$,
the second condition is vacuously satisfied. Note that Theorem \ref{prop:ident-given-a}
would fail in this case because neither state separates types.
\end{example}
\begin{example}
The following is an instance where the second condition of Theorem
\ref{prop:gen-indep-cols-iff} does not hold: 
\[
\alpha\left[\begin{array}{cccc}
1 & 0 & 1 & 0\\
0 & 1 & 0 & 1\\
0 & 0 & 0 & 0\\
0 & 0 & 0 & 0
\end{array}\right]+\beta\left[\begin{array}{cccc}
0 & 0 & 0 & 0\\
0 & 0 & 0 & 0\\
1 & 1 & 0 & 0\\
0 & 0 & 1 & 1
\end{array}\right]=\left[\begin{array}{cccc}
\alpha & 0 & \alpha & 0\\
0 & \alpha & 0 & \alpha\\
\beta & \beta & 0 & 0\\
0 & 0 & \beta & \beta
\end{array}\right],
\]
where $X=\left\{ x,y,z,w\right\} $. The columns of the right-hand
side matrix are linearly dependent for all $\alpha$ and $\beta$,
because the sum of the first and fourth columns equal to the sum of
the second and third column. To see how this case is ruled out in
Theorem \ref{prop:gen-indep-cols-iff}, note that there is one state-matching
$(m,\gamma)=\left(xzwy,abba\right)$, with $\Gamma(m,\gamma)=\left[2\ 2\right]$,
and another state-matching $(m',\gamma')=\left(zyxw,baab\right)$
with $\Gamma(m',\gamma')=\left[2\ 2\right]$. Observe that $m'$ is
obtained from $m$ by three swaps. It is easy to check tha there does
not exist any other state-matching. Hence, the second condition of
Theorem \ref{prop:gen-indep-cols-iff} is not satisfied.
\end{example}
Clearly, Theorem \ref{prop:ident-given-a} follows from Theorem \ref{prop:gen-indep-cols-iff},
since type separation by some $a^{*}\in A$ implies that there is
a state-matching such that $\gamma(t)=a^{*}$ for all $t\in T$, with
no other state-matchings having the same $\Gamma$.

It is also instructive to examine the following simplest $2\times2\times2$
case:
\begin{example}
\label{ex:2x2x2}Consider the type--state framework with two types
$T=\{t_{1},t_{2}\}$, two states $A=\{a,b\}$, and two alternatives
$X=\{x_{1},x_{2}\}$. Suppose, towards a violation of the identifiability
condition, that both states pool types---for instance, $c_{t_{1}}\left(a\right)=x=c_{t_{2}}\left(a\right)$
and $c_{t_{1}}\left(b\right)=y=c_{t_{2}}\left(b\right)$. Then there
exist two symmetric state-matchings: $\left(m,\gamma\right)=\left(xy,ab\right)$
and $\left(m',\gamma'\right)=\left(yx,ba\right)$. In both state-matchings,
each state is used exactly once and $m'$ is obtained from $m$ by
a single swap, violating the second condition of Theorem \ref{prop:gen-indep-cols-iff}.
This shows that at least one state must separate types for identifiability
to hold. What is more, this condition is also sufficient, since if
at least one state separates types, then Theorem \ref{Th:fully-id-given-a}
applies. 

These matching statements rest on underlying determinant conditions:
if state $s$ separates types, then $M^{s}$ has full rank (it is
the identity or a permutation matrix) and $\det M^{s}\neq0$, and
if state $s$ pools types, then $M^{s}$ has rank one (two identical
columns) and $\det M^{s}\neq0$. Global identification on the interior
of the simplex means that\footnote{Denoting $M\left(f\left(a\right)\right)=f\left(a\right)M^{a}+\left(1-f\left(a\right)\right)M^{b}$.}
$\det M(f(a))\neq0$ for all $f(a)\in(0,1)$. Since $\det M\left(f\left(a\right)\right)$
is a linear function of $f\left(a\right)$ in the $2\times2\times2$
case, this is equivalent to: $\det M^{a}$ and $\det M^{b}$ have
the same sign, and are not both zero. If exactly one state, say $a$,
separates types, then $\det M(f\left(a\right))=f(a)\det M^{a}$, which
is nonzero for all $f(a)\in(0,1)$. Hence the model is generically
(and indeed globally) identifiable. 

If both states separate types, two cases arise. When each type chooses
identically across states, $M^{a}=M^{b}=M$ and $M(f\left(a\right))$
is constant in $f\left(a\right)$ and invertible, so generic (and
indeed global) identification holds. By contrast, when each type chooses
differently across states, say $M^{a}$ is the identity and $M^{b}$
is a permutation matrix, $\det M(f\left(a\right))=2f\left(a\right)-1$
which vanishes at a unique interior value of $f(a)$, ensuring generic
(though not global) identification. The proof of Theorem \ref{cor:global_id_typestate}
generalizes this logic.
\end{example}
The reasoning in Example \ref{ex:2x2x2} shows as a corollary of Theorem
\ref{prop:gen-indep-cols-iff} that in the $2\times2\times2$ case
the parameters are generically identifiable if and only if there exists
at least one state that separates types. In subsection \ref{subsec:Identification-of-the}
we prove the result and its generalizations by means of nullspace
conditions.

\subsubsection{Global identifiability}

The characterizing condition for global identifiability closely mirrors
that of Theorem \ref{thm:full-id-SNS} in the general framework, with
one notable difference. In this context, \textquotedblleft global
identifiability\textquotedblright{} is understood with reference to
the parameter space in which state distributions are \emph{strictly
positive}, $f\in\left(0,1\right)^{\left|A\right|}$. Allowing $f$
to assign zero weight to some states would impose the additional requirement
that \emph{each} state-specific matrix $M^{a}$ be invertible. This
would trivialize the type--state framework, as it would imply that
types are already perfectly separated within each state. The matching
conditions in the following result therefore characterize global identifiability
on the interior of the simplex.
\begin{thm}
\label{cor:global_id_typestate}The parameters in $\pi$ are globally
identifiable on the interior of the simplex if and only if there exists
a subset $R\subseteq X$ with $|R|=r$ such that:

(i) $\mathcal{S}_{R}\neq\varnothing$; 

(ii) For any two state-matchings $(m,\gamma),(m',\gamma')\in\mathcal{S}_{R}$
, $m'$ can be obtained from $m$ via an even number of swaps. 
\end{thm}
\begin{proof}
Please refer to Appendix \ref{App:Cor_global_id_typestate}.
\end{proof}
The $2\times2\times2$ setting neatly illustrates the kind of different
behavioral restrictions that are imposed by generic and global identifiability
conditions. The following example is to be contrasted with Example
\ref{ex:2x2x2}:
\begin{example}
\label{ex:2x2x2 Global}Consider again the type--state framework
with two types $T=\{t_{1},t_{2}\}$, two states $A=\{a,b\}$, and
two alternatives $X=\{x_{1},x_{2}\}$. Let $f(a)\in(0,1)$. Note that
here necessarily $R=A$ if an $R$ as in the conditions of Theorem
\ref{cor:global_id_typestate} exists. We know from Example \ref{ex:2x2x2}
that type separation in at least one state is necessary for global
identifiability (since it is such for generic identifiability).

Suppose \emph{exactly} one state separates types---for instance,
$c_{t_{1}}\left(a\right)=x$, $c_{t_{2}}\left(a\right)=y$ and $c_{t_{1}}\left(b\right)=y=c_{t_{2}}\left(b\right)$.
Then there is exactly one state-matching, $\left(m,\gamma\right)=\left(xy,aa\right)$,
that uses only one state; and exactly one state-matching, $\left(m,\gamma'\right)=\left(xy,ab\right)$
that uses each state once. The conditions of Theorem \ref{cor:global_id_typestate}
are satisfied and there is global identification.

Suppose alternatively that both states separate types. Here there
is a crucial distinction, between the case in which \emph{each type
chooses in the same way} \emph{in the two states} (e.g. $c_{t_{1}}\left(s\right)=x$,
$c_{t_{2}}\left(s\right)=y$ for $s\in A$), and the case in which
\emph{each type} \emph{swaps choice from one state to the other} (e.g.
$c_{t_{1}}\left(a\right)=x=c_{t_{2}}\left(b\right)$, $c_{t_{1}}\left(b\right)=y=c_{t_{2}}\left(a\right)$).
It is clear that while in the first case there is a unique matching
and thus global identification, in the latter case we can construct
multiple matchings differing by one swap, violating condition (ii)
of Theorem \ref{cor:global_id_typestate} and unraveling global identification.
The determinant analysis of Example \ref{ex:2x2x2} underlies this
matching analysis.
\end{example}
The upshot of Example \ref{ex:2x2x2} is that in the $2\times2\times2$
case the parameters are globally identifiable if and only if either
there is exactly one state that separates types, or both states separate
the types but each type chooses identically across states. If each
type chooses differently across states, then identification can only
be generic.

Finally, we have seen (Theorem \ref{cor:global_id_typestate}) that
the existence of a state that separates types is sufficient for generic
identification. This begs the question: what additional condition
ensures global identification? A corollary to Theorem \ref{cor:global_id_typestate}
yields the answer. Fix a ``reference'' state $a^{*}\in A$ that
separates types. Define a binary relation $\leadsto$ (relative to
$a^{*}$) on $T$ by 
\[
t\leadsto t'\quad\Longleftrightarrow\quad\exists a\in A\ \text{such that}\ c_{t}(a)=c_{t'}(a^{*}).
\]

This simply means that type $t$ is, in some state, possible at the
alternative chosen by type $t'$ in the reference state.

We say that a bijection $\varphi:T\to T$ is a \emph{possible reassignment}
(relative to $a^{*}$) if 
\[
\forall t\in T,\quad t\leadsto\varphi(t).
\]

In words: each type $t$ is assigned a ``target'' type $t'$. Each
type chooses, in some state(s), an alternative chosen by its ``target''
type in the reference state, with no two types choosing the same alternative
(because $\varphi$ is bijective).
\begin{cor}
\label{Th: global_typestate_sufficient}Suppose there exists a state
$a^{*}\in A$ that separates types. Then the parameters $\pi$ are
globally identifiable on the interior of the simplex if and only if
every possible reassignment $\varphi$ is an even permutation of $T$.
\end{cor}
\begin{proof}
Please refer to Appendix \ref{App:Cor_global_id_typestate}.
\end{proof}

\subsection{Identification of the type distribution via nullspace conditions}\label{subsec:Identification-of-the}

The previous analysis provides powerful and general characterizations
of generic identification. However, those conditions are not completely
informative on how different states interact to produce or destroy
identification; and on the precise behavioral implications of such
failures. This section complements that analysis by adopting a geometric,
linear-algebraic perspective rather than a combinatorial one. We derive
nullspace conditions that make explicit the directions in aggregate
choice space along which identification fails, and we interpret lack
of identification as the existence of nontrivial directions in aggregate
choice space that can be generated in multiple ways.

This is transparently translated to behavioral conditions. The key
conceptual step is the notion of ``typical splits'', introduced
in Definition \ref{typical} below. Identification fails precisely
when distinct subsets of types generate the same aggregate choice
splits---that is, the same differences in choice frequencies between
two alternatives---so that these directions cannot be uniquely attributed
to individual types. 

For brevity and to obtain clear conditions, we focus on the case of
two states, and on generic identification only.

\subsubsection{The $2\times2$ and $3\times3$ cases}

We proceed in stages, picking up from the simple case of two states,
two types, and two alternatives from the end of the previous section.
In this case, Equation \eqref{eq:matrix-conv-comb} reduces to the
convex combination of two $2\times2$ matrices, where generic identification
can be characterized as follows.
\begin{thm}
\label{nullspace2by2} For $p=M\pi$, where $M=\sum_{a\in A}f(a)M^{a}$
is a $2\times2$ matrix as defined in Equation \eqref{eq:matrix-conv-comb}
with $A=\{a,b\}$ and $T=\{t_{1},t_{2}\}$, the following statements
are equivalent 
\end{thm}
\begin{itemize}
\item[(i)]  The parameters in $\pi$ are generically identifiable
\item[(ii)] $\operatorname{Null}(M^{a})\cap\operatorname{Null}(M^{b})=\{{\bf 0}\}$ 
\item[(iii)] $c_{t_{1}}(a)\neq c_{t_{2}}(a)$ or $c_{t_{1}}(b)\neq c_{t_{2}}(b)$ 
\end{itemize}
\begin{proof}
Please refer to Appendix \ref{app:4}.
\end{proof}
This first result via nullspace conditions establishes that the parameters
in $\pi$ are generically identifiable if and only if there is a state
that separates types (condition (iii) of Theorem \ref{nullspace2by2}).
On the other hand, the theorem also provides a purely algebraic characterization
of generic identification. The equivalence of Condition (i) and (ii)
shows that the parameters in $\pi$ are generically identifiable if
and only if the intersection of the nullspaces of $M^{a}$ and $M^{b}$
is the trivial one, i.e., the only ${\bf x}$ satisfying both $M^{a}{\bf x}={\bf 0}$
and $M^{b}{\bf x}={\bf 0}$ is the zero-vector. Given the specific
one-hot structure of these matrices, this nullspace condition is equivalent
to restricting one of the individual nullspaces to the trivial one,
i.e., $\operatorname{Null}(M^{a})=\{{\bf 0}\}$ or $\operatorname{Null}(M^{b})=\{{\bf 0}\}$.
This demonstrates the equivalence of Condition (ii) and (iii).

Next, we turn to the case of two states, three types, and three alternatives.
As in the previous case, it remains necessary for identification that
not all states pool types. But, this condition alone is \emph{not}
sufficient for generic identification in the $3\times3$ case. The
next example illustrates why.
\begin{example}
\label{example_3by3} This example consists of two parts. This first
part motivates why at least one state must separate types for identifiability
when there are three alternatives and three types. Let $A=\{a,b\}$
be the state space, $T=\{t_{1},t_{2},t_{3}\}$ the type space, and
$p=M\pi$, where $M=f(a)M^{a}+f(b)M^{b}$, for all $f(a),f(b)\in[0,1]$,
$f(a)+f(b)=1$, is a $3\times3$ matrix. Now suppose that both states
pool types, i.e., let $c_{t_{1}}(a)=c_{t_{2}}(a)$ and $c_{t_{1}}(b)=c_{t_{2}}(b)$.
For instance, let $M$ be given as follows 
\[
f(a)\underbrace{\left[\begin{array}{ccc}
1 & 1 & 0\\
0 & 0 & 1\\
0 & 0 & 0
\end{array}\right]}_{M^{a}}+f(b)\underbrace{\left[\begin{array}{ccc}
0 & 0 & 0\\
1 & 1 & 0\\
0 & 0 & 1
\end{array}\right]}_{M^{b}}=\underbrace{\left[\begin{array}{ccc}
f(a) & f(a) & 0\\
f(b) & f(b) & f(a)\\
0 & 0 & f(b)
\end{array}\right]}_{M}
\]
Then, the parameters in $\pi$ are not identifiable. To see this observe
that $M$ is clearly not invertible, because $\operatorname{Null}(M)=\operatorname{Null}(M^{a})\cap\operatorname{Null}(M^{b})\neq\{{\bf 0}\}$.

The next part shows that the existence of a type separating state
does not by itself guarantee generic identification. To see this,
let the matrix $M$ be given as follows 
\[
f(a)\underbrace{\left[\begin{array}{ccc}
1 & 1 & 0\\
0 & 0 & 1\\
0 & 0 & 0
\end{array}\right]}_{M^{a}}+f(b)\underbrace{\left[\begin{array}{ccc}
0 & 1 & 1\\
1 & 0 & 0\\
0 & 0 & 0
\end{array}\right]}_{M^{b}}=\underbrace{\left[\begin{array}{ccc}
f(a) & 1 & f(b)\\
f(b) & 0 & f(a)\\
0 & 0 & 0
\end{array}\right]}_{M}
\]
Note that no two columns of $M$ coincide for all $f(a)$ and $f(b)$.
In other words, for any two types, there exists a state that separates
them. Nevertheless, the parameters in $\pi$ are not identifiable.
Indeed, $M$ is not invertible, because it has a zero row indicating
that the corresponding alternative is never chosen. What captures
this algebraically is the fact that the nullspaces of the \emph{transposes
}of $M^{a}$ and $M^{b}$ have the following non-trivial intersection,
\[
\operatorname{Null}(M^{T})=\operatorname{Null}(M^{{a}^{T}})\cap\operatorname{Null}(M^{{b}^{T}})=\operatorname{span}\left\{ \begin{bmatrix}0\\
0\\
1
\end{bmatrix}\right\} .
\]
\end{example}
Hence, for generic identification in the $3\times3$ case some additional
condition is needed besides not all states pooling types. The additional
requirement is that \emph{aggregate choices spread over all three
alternatives}. The next result shows that, taken together, these two
conditions characterize generic identification in the $3\times3$
case.
\begin{thm}
\label{nullspace3by3} For $p=M\pi$, where $M=\sum_{a\in A}f(a)M^{a}$
is a $3\times3$ matrix as defined in Equation \eqref{eq:matrix-conv-comb}
with $A=\{a,b\}$, $X=\{x,y,z\}$, and $T=\{t_{1},t_{2},t_{3}\}$,
the following are equivalent 
\end{thm}
\begin{itemize}
\item[(i)]  The parameters in $\pi$ are generically identifiable
\item[(ii)] $\operatorname{Null}(M^{a})\cap\operatorname{Null}(M^{b})=\{{\bf 0}\}\quad\text{ and }\quad\operatorname{Null}(M^{{a}^{T}})\cap\operatorname{Null}(M^{{b}^{T}})=\{{\bf 0}\}$ 
\item[(iii)] For all $t,t'\in T$ with $t\neq t'$, there exists an $a\in A$
such that $c_{t}(a)\neq c_{t'}(a)$; and, for all $w\in X$, there
exist $a\in A$ and $t\in T$ such that $c_{t}(a)=w$ 
\end{itemize}
\begin{proof}
Please refer to Appendix \ref{app:5}.
\end{proof}
The equivalence between Condition (i) and (iii) establishes that,
in the $3\times3$ case, generic identification is equivalent to two
requirements. The first is the existence of a type separating state.
The second is that none of the three alternatives is never chosen.
The theorem also shows how this translates in terms of nullspaces.
The equivalence between Condition (i) and (ii) shows that the parameters
in $\pi$ are generically identifiable if and only if \emph{both}
the intersection of the nullspaces of $M^{a}$ and $M^{b}$ and the
intersection of the nullspaces of $M^{{a}^{T}}$ and $M^{{b}^{T}}$
are the trivial ones. Given the one-hot structure of these matrices,
$\operatorname{Null}(M^{a})\cap\operatorname{Null}(M^{b})=\{{\bf 0}\}$
is equivalent to the existence of a type separating state, and $\operatorname{Null}(M^{{a}^{T}})\cap\operatorname{Null}(M^{{b}^{T}})=\{{\bf 0}\}$
is equivalent to none of the three alternatives being never chosen.

With three alternatives and three types, the aggregate choice shares
may allow to identify all types, \emph{even though choice shares from
any single state alone may not do so}. Neither the data in $M^{a}$
alone, nor the data in $M^{b}$ alone would suffice for identifying
the parameters in $\pi$, but together they guarantee generic identification.
The next example illustrates such a situation.
\begin{example}
Let $A=\{a,b\}$ be the state space and let there be three types and
three alternatives, i.e., $|T|=|X|=3$. Consider $p=M\pi$, where
the $3\times3$ matrix $M=f(a)M^{a}+f(b)M^{b}$, for all $f(a),f(b)\in[0,1]$,
$f(a)+f(b)=1$, is given as follows 
\[
f(a)\underbrace{\left[\begin{array}{ccc}
1 & 1 & 0\\
0 & 0 & 1\\
0 & 0 & 0
\end{array}\right]}_{M^{a}}+f(b)\underbrace{\left[\begin{array}{ccc}
0 & 0 & 0\\
1 & 0 & 0\\
0 & 1 & 1
\end{array}\right]}_{M^{b}}=\underbrace{\left[\begin{array}{ccc}
f(a) & f(a) & 0\\
f(b) & 0 & f(a)\\
0 & f(b) & f(b)
\end{array}\right]}_{M}
\]
Neither the data in $M^{a}$ alone nor that in $M^{b}$ alone suffices
for identification, as each such matrix has a zero row. Taken together,
however, this data allows for generic identification. For any two
types, there exists a state that separates them. Furthermore, $M$
has no zero row which shows that aggregate choices spread over all
three alternatives. In other words, neither $M^{a}$ nor $M^{b}$
are invertible. But, since $\operatorname{Null}(M^{a})\cap\operatorname{Null}(M^{b})=\{{\bf 0}\}$
and $\quad\operatorname{Null}(M^{{a}^{T}})\cap\operatorname{Null}(M^{{b}^{T}})=\{{\bf 0}\}$,
Theorem \ref{nullspace3by3} implies that their convex combination
$M$ leads to $\pi$ being generically identifiable from $p=M\pi$.
Indeed, since $\det(M)=-f(a)f(b)(f(a)+f(b))$, the parameters in $\pi$
are generically identifiable as long as neither $f(a)$ nor $f(b)$
are zero. 
\end{example}

\subsubsection{The general $r\times r$ case}

The previous results suggest that a single type separating state may
suffice for generic identification. Indeed, in line with Theorem \ref{prop:ident-given-a},
this is true in the most general $r\times r$ case. That is, it holds
for any $r\in\mathbb{N}$ and any finite number of $r\times r$ matrices.
But the result can be further strengthened. For the parameters in
$\pi$ to be generically identifiable from $p=M\pi$ and $M=\sum_{a\in A}M^{a}$,
it suffices that there exists some subset $B\subseteq A$ such that
the parameters in $\pi$ would be generically identifiable from $p=N\pi$
and $N=\sum_{b\in B}f(b)M^{b}$. The next result spells out this fact.
\begin{thm}
\label{invertiblegeneric} If the parameters in $\pi$ are generically
identifiable from $p=N\pi$ where $N=\sum_{b\in B}f(b)M^{b}$ is as
defined in Equation \eqref{eq:matrix-conv-comb}, then they continue
to be generically identifiable from $p=M\pi$ , where $M=\alpha M^{a}+(1-\alpha)N$. 
\end{thm}
\begin{proof}
Please refer to Appendix \ref{app:6}.
\end{proof}
The implication of this result for identification is the following:
if data from some subset of states suffices for the parameters in
$\pi$ to be generically identifiable, then the inclusion of additional
states cannot undo this identification. A simple corollary is that,
within the type--state framework, the parameters in $\pi$ are generically
identifiable if and only if there exists at least a parameter configuration
that is unambiguously revealed by the data. The next result formalizes
this observation.
\begin{cor}
\label{invertiblethengenericinvertible} The parameters in $\pi$
are generically identifiable from $p=M\pi$, where $M=\sum_{a\in A}f(a)M^{a}$
is as defined in Equation \eqref{eq:matrix-conv-comb}, if and only
if there exists some probability measure $\overline{f}$ on $A$ such
that $p=\overline{M}\pi$ and $p\neq\overline{M}\pi'$ for all $\pi'\neq\pi$,
where $\overline{M}=\sum\overline{f}(a)M^{a}$. 
\end{cor}
\begin{proof}
Please refer to Appendix \ref{app:7}.
\end{proof}
This result implies in particular that if the parameters in $\pi$
are\emph{ not} generically identifiable, then there does not exist
\emph{any} parameter configuration that is uniquely revealed by the
data $p$. While lack of generic identification is defined as the
failure of identification on a set of parameter values of nonzero
Lebesgue measure, Corollary \ref{invertiblethengenericinvertible}
implies a stronger conclusion, already encountered in the matching
approach: \emph{when generic identification of $\pi$ fails, it fails
structurally.} 

Finally, we extend the conditions from Theorem \ref{nullspace2by2}
and Theorem \ref{nullspace3by3} to the general $r\times r$ case.
The ensuing theorem specifies the form of qualitative behavioral heterogeneity
that is both necessary and sufficient for generic identification.
To state it precisely, we first introduce a key definition, formalizing
when a given pattern of aggregate choice shares can be regarded as
typical of a group of types. 
\begin{defn}
\label{typical} Given $M$ as defined in Equation \eqref{eq:matrix-conv-comb}
and some non-empty set of types $S\subseteq T$, the split between
two alternatives $x,y\in X$ is \emph{typical} \emph{of} $S$ if 
\[
({\bf e}^{x}-{\bf e}^{y})\in\operatorname{span}\{M_{*j}:\text{column }j\text{ corresponds to some type }t\in S\}
\]
where for $w\in\{x,y\}$, ${\bf e}^{w}$ is defined as follows: for
all $i=1,\dots,r$ 
\[
{\bf e}_{i}^{w}=\begin{cases}
1, & \text{if Row \ensuremath{i} corresponds to alternative \ensuremath{w}}\\
0, & \text{otherwise}
\end{cases}
\]
\end{defn}
Intuitively, a split between two alternatives is typical of a set
of types if the aggregate difference in their choice frequencies between
these alternatives can be replicated as a linear combination of the
choice patterns generated by those types alone. In this case, observed
differences in frequency between the two alternatives do not provide
information that helps distinguish a particular type within the set,
as the same aggregate split could be generated by reallocating choice
probabilities among the remaining types. Typical splits therefore
capture directions in aggregate choice space that are not informative
about the behavior of individual types, and play a central role in
determining when different patterns of behavior can or cannot be uniquely
attributed to distinct types.

The next example illustrates the idea of a split between two alternatives
being typical of some set of types.
\begin{example}
Let $A=\{a,b\}$ be the state space and let there be four types, $T=\{1,2,3,4\}$
and four alternatives, $X=\{w,x,y,z\}$. Further, let the $4\times4$
matrix $M=f(a)M^{a}+f(b)M^{b}$, for all $f(a),f(b)\in[0,1]$, $f(a)+f(b)=1$,
be given as follows 
\[
f(a)\underbrace{\begin{bmatrix}1 & 1 & 0 & 0\\
0 & 0 & 1 & 0\\
0 & 0 & 0 & 1\\
0 & 0 & 0 & 0
\end{bmatrix}}_{M^{a}}+f(b)\underbrace{\begin{bmatrix}0 & 0 & 0 & 0\\
1 & 0 & 0 & 0\\
0 & 1 & 0 & 0\\
0 & 0 & 1 & 1
\end{bmatrix}}_{M^{b}}=\underbrace{\begin{bmatrix}f(a) & f(a) & 0 & 0\\
f(b) & 0 & f(a) & 0\\
0 & f(b) & 0 & f(a)\\
0 & 0 & f(b) & f(b)
\end{bmatrix}}_{M}
\]
where type $t\in T$ corresponds to the $t$-th Column of $M$, and
alternative $w$ to Row 1, $x$ to Row 2, $y$ to Row 3, and $z$
to Row 4. Note that $\det(M)=0$. Now, for $f(a),f(b)\neq0$, the
split between alternatives $x$, $y$ is typical of the set of types
$\{3,4\}$, as 
\[
\begin{bmatrix}\;\;\,0\\
\;\;\,1\\
-1\\
\;\;\,0
\end{bmatrix}=\frac{1}{f(a)}(M_{*3}-M_{*4})
\]
But this same split is also typical of the set of types $\{1,2\}$.
To see this, observe that $f(a)(M_{*1}-M_{*2})=f(b)(M_{*3}-M_{*4})$.
\end{example}
The definition of a split between two alternatives being typical of
a set of types provides the behavioral notion of heterogeneity that
is both necessary and sufficient for identification. In the previous
example, the fact that $\det(M)=0$ indicates that the same split
cannot be typical for distinct sets of types. The next theorem formalizes
this observation and establishes an equivalent algebraic condition,
expressed in terms of non-trivial intersections between the nullspaces
of the relevant matrices.
\begin{thm}
\label{nullspaces_general} For an $r\times r$ matrix $M=\sum_{a\in A}f(a)M^{a}$
as defined in Equation \eqref{eq:matrix-conv-comb} with $A=\{a,b\}$,
$|X|=|T|=r$, the following are equivalent 
\end{thm}
\begin{itemize}
\item[(i)]  The parameters in $\pi$ are generically identifiable
\item[(ii)]  For almost all $\{f(a)\}_{a\in A}$, for all non-trivial ${\bf x}\in\operatorname{Null}(M^{a})$,
all non-trivial ${\bf y}\in\operatorname{Null}(M^{b})$, and all ${\bf z}\in\operatorname{span}\{{\bf x},{\bf y}\}^{\perp}$,
it holds that 
\[
(f(a)M^{a}+f(b)M^{b}){\bf z}\neq M^{a}{\bf y}+M^{b}{\bf x}
\]
\item[(iii)]  There exists $a,b\in A$, $a\neq b$, such that, for all $t,t'\in T$,
$t\neq t'$, with $c_{t}(a)=c_{t'}(a)$, it holds that $c_{t}(b)\neq c_{t'}(b)$
and the split between $c_{t}(b)$ and $c_{t'}(b)$ is not typical
of $(T\setminus\{t,t'\})$ 
\end{itemize}
\begin{proof}
Please refer to Appendix \ref{app:8}.
\end{proof}
Condition (iii) conveys two key implications. First, for any two types,
there must exist at least one state that separates them. Second, if
two types are pooled by one state, then the way in which their choices
differ in another state cannot itself be typical of the choices of
all other types. The condition connects directly to the identification-by-matching
results from the previous section. If two types are pooled every state,
they cannot be matched to distinct alternatives. Moreover, if two
types are pooled by one state and are separated by another, but their
difference in choice is also typical of another subset of types, then
the uniqueness requirement of the matching result in Theorem \ref{prop:gen-indep-cols-iff}
fails to hold.

To understand Condition (ii), first observe that Theorem \ref{invertiblegeneric}
implies that the parameters in $\pi$ are generically identifiable
whenever there is no non-trivial ${\bf x}\in\operatorname{Null}(M^{a})$
or no non-trivial ${\bf y}\in\operatorname{Null}(M^{b})$. It follows
that any non-trivial ${\bf w}\in\operatorname{Null}(M)$ can be written
as a linear combination involving some such non-trivial ${\bf x}$
and ${\bf y}$. That is, any such non-trivial ${\bf w}$ can be decomposed
as follows 
\begin{equation}
{\bf w}={\bf x}+{\bf y}+{\bf z}\label{wxyz}
\end{equation}
for some non-trivial ${\bf x}\in\operatorname{Null}(M^{a})$, some
non-trivial ${\bf y}\in\operatorname{Null}(M^{b})$ and some ${\bf z}\in\operatorname{span}\{{\bf x},{\bf y}\}^{\perp}$.
Now, the parameters in $\pi$ being generically identifiable implies
that, for almost all $\{f(a)\}_{a\in A}$, there exists no such non-trivial
${\bf w}\in\operatorname{Null}(f(a)M^{a}+f(b)M^{b})$. In other words,
for almost all $\{f(a)\}_{a\in A}$, all ${\bf w}$ from Equation
\eqref{wxyz} satisfy $(f(a)M^{a}+f(b)M^{b}){\bf w}\neq{\bf 0}$.
As the proof of Theorem \ref{nullspaces_general} establishes this
inequality is equivalent to the one in Condition (ii).

Condition (ii) furthermore generalizes the conditions of the previous
identification theorems in this section. First, observe that Condition
(ii) generalizes Theorem \ref{invertiblegeneric} in the following
way: If $\pi$ is globally identifiable from $p=M^{a}\pi$ (or $p=M^{b}\pi$),
the nullspace of $M^{a}$ (or $M^{b}$) is the trivial one, i.e.,
it only contains the zero vector. Since Condition (ii) only refers
to non-trivial elements of the nullspaces of $M^{a}$ and $M^{b}$,
in line with Theorem \ref{invertiblegeneric}, Condition (ii) is then
satisfied automatically. Next, observe that $\operatorname{span}\{{\bf x},{\bf y}\}^{\perp}$
always contains the zero vector. So, Condition (ii) rules out that
the nullspaces of $M^{a}$ and $M^{b}$ intersect in a non-trivial
way. To see this, suppose otherwise and let there be some ${\bf x}\neq{\bf 0}$
which is an element of both $\operatorname{Null}(M^{a})$ and $\operatorname{Null}(M^{b})$.
But, then ${\bf 0}=f(a)M^{a}{\bf y}+f(b)M^{b}{\bf x}$, contradicting
Condition (ii), as ${\bf 0}\in\operatorname{span}\{{\bf x},{\bf y}\}^{\perp}$.
As such, Condition (ii) generalizes the corresponding condition in
Theorem \ref{nullspace2by2} from the $2\times2$ case. Since $\operatorname{span}\{{\bf x},{\bf y}\}^{\perp}$
contains the zero vector, another way in which Condition (ii) may
fail is when $M^{a}{\bf y}=M^{b}{\bf x}$ holds, for some non-trivial
${\bf x}\in\operatorname{Null}(M^{a})$ and some non-trivial ${\bf y}\in\operatorname{Null}(M^{b})$.
Therefore, Condition (ii) also generalizes the corresponding condition
in Theorem \ref{nullspace3by3} from the $3\times3$ case.\footnote{We should also note that, although parts of the proof make use of
linear decompositions of vectors into components lying in different
subspaces, no inner-product structure or notion of orthogonality plays
any \emph{substantive} role in the argument. All results depend only
the one-hot structure of the matrices $\ensuremath{M^{a}}$ and $M^{b}$,
and on the resulting row-sum identities that characterize typical
splits. Any reference to complementary subspaces is purely a notational
convenience (and could in fact be replaced by an arbitrary linear
complement). In particular, the identification results are invariant
to the choice of inner product and do not rely on any metric properties
of $\mathbb{R}^{r}$ . The characterization remains in this sense
entirely qualitative in nature.}
\begin{rem}
All characterization results in this section can be generalized to
the case of an $n\times r$ matrix $M$, where $n\geq r$. To this
end, the corresponding statements need to be adjusted to saying that
the matrix M contains an $r\times r$ submatrix to which these statements
then apply.
\end{rem}

\section{Multiple choice occasions and the identification of type-level choice
probabilities}\label{sec:Inference-from-multiple}

This section advances the analysis by studying how identification
can be extended to type-level choice probabilities. To this end, we
consider an analyst who observes choice data from a finite number
of occasions $j\in\left\{ 1,\ldots,J\right\} $. The extreme case
of $J=1$ is that of a one-shot observation, which was the focus of
the previous sections. We assume that the distribution $\pi$ of types
remains constant across occasions. For each occasion $j$, agents
choose from a feasible set $X_{j}.$ We assume that choices are independent
across occasions conditional on the type. The notation and definitions
introduced earlier extend in the natural way to all $j$-indexed objects
of this section. 

The data set is now described as a tensor\footnote{A tensor is a multidimensional array that extends the concept of a
matrix to an arbitrary number of indices, known as the tensor\textquoteright s
order. Its dimensions specify how many values each index can take,
generalizing the numbers of rows and columns in a matrix.} $S$ of order $J$ with dimensions $n_{1}\times\ldots\times n_{J}$.
A typical entry $S_{k_{1},\ldots,k_{J}}$ is the share of the population
choosing alternative $x_{k_{j}}$ from $X_{j}$ for $j\in\left\{ 1,\ldots,J\right\} $.
Hence, $S$ contains the joint distribution of choices for the $J$
occasions. We can express $S$ using $M_{j}$ and $\pi$ as 
\begin{equation}
S=\sum_{h=1}^{r}\pi(h)\bigotimes_{j=1}^{J}M_{j}\boldsymbol{1}_{h},\label{eq:choice-tensor-1-1}
\end{equation}
where $\bigotimes$ is the outer product operator and $\boldsymbol{1}_{h}$
is the unit vector for component $h$.\footnote{Thus, $M_{j}\boldsymbol{1}_{h}$ is the $h$th column of $M_{j}$.
The outer product operation then creates a rank-1 tensor of order
$J$. A tensor is said to be of rank 1 if it is an outer product of
vectors. Therefore, $S$ is a convex combination of $r$ rank-1 tensors. }

It turns out that $J=3$ is sufficient for the generic identification
of the model in Equation \eqref{eq:choice-tensor-1-1}, provided the
possibility relations are sufficiently rich. Under this condition,
one can apply algebraic results on tensor decompositions to identify
the model. The uniqueness properties of these decompositions have
been extensively studied, beginning with Kruskal\textquoteright s
\cite{Kruskal1977} fundamental theorem and later refined by Allman,
Matias, and Rhodes \cite{AllmanMatiasRhodes2009} and others. Because
the spirit of our analysis is that of identifying the parameters from
minimal data, we restrict the agent\textquoteright s choices to the
three occasions required. Additional occasions neither aid nor impede
identification in this setting, whereas two occasions are insufficient.

\subsection{General framework}

In the general case, the parameter space of the model is a full dimensional
subset of $\left[0,1\right]^{m}$ where $m=(r-1)+\sum_{j=1}^{3}(\left|\land_{j}\right|-r)$.
The model's parameterization map takes values in $\left[0,1\right]^{n_{1}\times n_{2}\times n_{3}}$. 

The next result gives a sufficient condition for generic identification.
\begin{thm}
Suppose that there exist natural numbers $v_{j}\leq r$ for $j\in\left\{ 1,2,3\right\} $
such that:

(i) $v_{1}+v_{2}+v_{3}\geq2r+2$.

(ii) For any $j\in\left\{ 1,2,3\right\} $ and any $v_{j}$ types
in $T$, there is a matching of these types with alternatives in $X_{j}$.

Then the parameters in $\pi,$ $M_{1}$, $M_{2}$, and $M_{3}$ are
generically identifiable, up to label swapping.\label{prop:ident-3occasions}
\end{thm}
\begin{proof}
Please refer to Appendix \ref{app:9}.
\end{proof}
Note that for the application of Theorem \ref{prop:ident-3occasions}
we need at least $2r+2$ alternatives across the three occasions that
have positive choice shares. Additionally, because identification
is unique up to a simultaneous permutation of the entries of $\pi$
and columns of $M_{1}$, $M_{2}$, and $M_{3}$, we do not know a
priori which of the inferred values of $\pi$ corresponds to a particular
type. However, the correct labeling is often revealed by the pattern
of 0's in $M_{j}$. The matrix in Example \ref{ex:M-matrix-3x3} may
serve to illustrate this.
\begin{example}
\label{example multiocc 1}Consider the following matrix form, with
rows corresponding to four alternatives $x$, $y$, $z$, and $w$,
and columns corresponding to types $t_{1}$, $t_{2}$, and $t_{3}$,\[
M_j = \bordermatrix{~ & t_1 & t_2 & t_3 \cr
               x & \color{red}{\bullet} & \bullet & \bullet \cr
               y & 0 & 0 & \color{red}{\bullet} \cr
               z & \bullet & \color{red}{\bullet} & \bullet \cr
               w & 0 & 0 & \bullet \cr}.
\]Note that $t\land_{j}x$ and $t\land_{j}z$ for all types $t$, but
$t\land_{j}y$ and $t\land_{j}w$ only for $t=t_{3}$. It is easy
to see that the second condition of Theorem \ref{prop:ident-3occasions}
is satisfied with $v_{j}=3$ (for example, take the matching represented
by the red (${\color{red}{\bullet}}$) dots).
\end{example}
Clearly, when a type-to-alternative matching as in Definition \ref{def:matching}
exists, it implies that $v_{j}=r$, making it easier to satisfy the
first condition of Theorem \ref{prop:ident-3occasions}. The next
examples illustrates a situation where less information is conveyed
by choices. 
\begin{example}
\label{example multiocc 2}Consider the following matrix form: \[
M_j = \bordermatrix{~ & t_1 & t_2 & t_3 \cr
               x & \bullet & \bullet & \bullet \cr
               y & 0 & 0 & \bullet \cr
               z & 0 & 0 & \bullet \cr}.
\] In this case, the second condition is only true for $v_{j}=1$ because
there is no matching covering both types $t_{1}$ and $t_{2}$. Intuitively,
this pattern of choices does not help to identify the shares of $t_{1}$
and $t_{2}$ because the agents of these types always choose $x$. 
\end{example}

\subsection{Type-state framework}

We now consider multiple occasions of choice in the type-state framework.
The states are allowed to vary with the occasions, with $A_{j}$ denoting
the set of states at occasion $j$. When the conditions of Theorem
\ref{prop:gen-indep-cols-iff} are satisfied for each occasion, we
can prove the generic identifiability of the parameters of our model.
\begin{thm}
Consider the model described by Equations \eqref{eq:matrix-conv-comb}
and \eqref{eq:choice-tensor-1-1} with parameters $\pi$ and $f_{j}$.
If, for any occasion $j\in\{1,2,3\}$, the conditions of Theorem \ref{prop:gen-indep-cols-iff}
hold, then $\pi$ and $\left\{ M_{j}:j\in\left\{ 1,2,3\right\} \right\} $
are generically identified, up to label swapping.\label{prop:ident-3occasions-type-state}
\end{thm}
\begin{proof}
Please refer to Appendix \ref{app:10}.
\end{proof}
Theorem \ref{prop:ident-3occasions-type-state} provides a sufficient
condition for the generic identification of $\pi$ and $M_{j}$, $j\in\left\{ 1,2,3\right\} $.
It is sometimes also possible to recover $f_{j}$ from Equation \eqref{eq:matrix-conv-comb}.
In particular, this occurs when there exists a type that selects a
different alternative in every state $a\in A_{j}$. More generally,
Equation \eqref{eq:matrix-conv-comb} defines a system of linear equations
in the unknowns $\left\{ f(a)\right\} _{a\in A_{j}}$, which may or
may not admit a unique solution.

Theorems \ref{prop:gen-indep-cols-iff} and \ref{prop:ident-3occasions-type-state}
can be generalized to the case where the state distribution $f_{j}$
depends on the type. In that case, instead of Equation \eqref{eq:matrix-conv-comb}
we would have that each column of $M_{j}$ is a different convex combination
of the corresponding columns of $\left\{ M^{a}:a\in A_{j}\right\} $.

\section{Applications}\label{sec:Applications}

In this section we illustrate concretely how our framework can be
applied, by studying specific environments.

\subsection{Choice over vectors of characteristics}

Suppose that consumer goods are described in terms of measurable characteristics
(after Lancaster \cite{Lancaster1966}), or \emph{features}. Denote
by $F=\{1,\dots,N\}$ the set of possible features.\footnote{\cite{KopsManziniMariottiPasichnichenko2022} studies a version of
this model using only \emph{deterministic} choices.} A good $x=\left(x_{1},...,x_{N}\right)$ is described by the amounts
$x_{i}$, for all features $i\in F$, of that good (while we use the
term ``good'' for simplicity, as explained below we allow features
to be undesirable). 

Consumers choose from feasible set $Y\subseteq\mathbb{R}^{N}$. We
assume that the feasible set $Y$ is a nonempty closed bounded convex
set of $\mathbb{R}^{N}$ with the boundary denoted by $\partial Y$.

A \textit{type }\textit{\emph{here is an}} $I\subseteq F$ corresponding
to the nonempty set of features that the consumer considers relevant
for the description of the alternatives. A consumer type $I$ chooses
an $x^{*}\in Y$ that is ``optimal'' with respect to the features
$I$. The next definition specifies the notion of optimality. 
\begin{defn}
Given an evaluation function $e:F\to\left\{ -1,0,1\right\} $, a point
$x\in Y$ is \textit{$e$-admissible} if $y\ensuremath{\notin Y}$
whenever $e\left(i\right)y_{i}\geq e\left(i\right)x_{i}$ for all
$i\in F$ and $e\left(i^{*}\right)y_{i^{*}}>e\left(i^{*}\right)x_{i^{*}}$
for some $i^{*}\in F$.
\end{defn}
The consumer's evaluation function $e$ indicates if the consumer
values a feature positively, negatively, or ignores it. Clearly, for
a consumer of type $I$, we impose that $e\left(i\right)\neq0$ if
and only if $i\in I$. The \textit{\emph{behavioral assumption}} is
that for each consumer there exists an (unobserved) evaluation function
$e$ such that the choice $x^{*}\in Y$ is \textit{$e$}-admissible\textit{.
}This requirement is similar to Pareto optimality; however, it is
applied only to the dimensions in $I$, and the directions of monotonicity
specified by $e$ are not a priori known to the analyst.
\begin{defn}
\label{def:type-possible-at-x}A type $I\subseteq F$ is \emph{possible}
at $x^{*}\in Y$ if there exists an evaluation function $e$ such
that $e(i)\neq0\iff i\in I$ and $x^{*}$ is \textit{$e$}-admissible.
\end{defn}
Let $X\subset Y$ be a finite set of goods with positive choice shares.
For each type $I\subseteq F$ and each good $x\in X$, we write $I\land x$
if type $I$ is possible at $x\in Y$ according to Definition \ref{def:type-possible-at-x}.

As an example, let $N=2$, so the types are $\left\{ 1\right\} $,
$\left\{ 2\right\} $, and $\left\{ 1,2\right\} $. The feasible set
$Y$ is shown in Figure~\ref{fig:rhombus}. As one can readily verify,
the possible types at each goods $X=\left\{ x,y,z,w\right\} $ coincide
with those described by $M_{j}$ in Example \ref{example multiocc 1}.
Hence, Theorem \ref{prop:generic-ident} guarantees that the distribution
of the types $\left\{ 1\right\} $, $\left\{ 2\right\} $, and $\left\{ 1,2\right\} $
can be generically recovered from the observed choices. With three
occasions, Theorem \ref{prop:ident-3occasions} further shows that
we can generically identify both the type distribution and the matrices
$M_{1}$, $M_{2}$, and $M_{3}$.

The type-state specialization of the general framework can be applied
to a population of consumer types that satisfy a linearity restriction.
Formally, let us say that a type $I\subseteq F$ is \emph{linear}
at $x^{*}\in Y$ if there exists an $a\in\mathbb{R}^{N}$ such that
$a_{i}\neq0\iff i\in I$ and $x^{*}\in\argmax_{x\in Y}\left\langle a,x\right\rangle $.
It is easy to show that if a type is linear at $x^{*}$, then it is
also possible at $x^{*}$. 

For this application, we assume that the feature coefficients are
drawn from a finite set of values. Thus, $f_{j}$ is interpreted as
a probability measure defined on a finite set $A_{j}\subset\mathbb{R}^{N}$
of feature values vectors. After a realization of $a$ according to
$f_{j}$, the consumer solves 
\begin{equation}
\max_{x\in Y_{j}}\left\langle \kappa_{I}(a),x\right\rangle \label{eq:agent's-problem-projected}
\end{equation}
where $\kappa_{I}:\mathbb{R}^{N}\to\mathbb{R}^{N}$ is defined as
$\kappa_{I}(a)_{i}=a_{i}$ if $i\in I$ and $\kappa_{I}(a){}_{i}=0$
otherwise.

We exclude the non-generic case when the consumer problem \eqref{eq:agent's-problem-projected}
does not have a unique solution.\footnote{Alternatively, a tie-breaker assumption can be used.}
Thus, Equation \eqref{eq:agent's-problem-projected} determines a
choice function in the sense previously specified, and the results
of Section \ref{sec:One-shot-choice Type-State} can be applied to
ensure that the distribution of consumer types can be identified.

Some additional identification analysis can be carried out in the
model: in fact, $M_{j}$ contains some information about the evaluation
functions in the consumer population. Let us consider the example
of Figure \ref{fig:rhombus} again. Recall that each consumer's choice
is $e$-admissible with respect to the consumer's evaluation function
$e$. The first column of $M_{j}$ describes the tastes of the consumers
for whom only feature $x_{1}$ matters. Precisely, $M_{j}(1,1)$ is
the share of $e(1)=1$ ($x_{1}$ is desirable) and $M_{j}(3,1)$ is
the share of $e(1)=-1$ ($x_{1}$ is undesirable). Likewise, the second
column describes the tastes of the type $\left\{ 2\right\} $ agents.
Finally, the third column describes the tastes of the type $\left\{ 1,2\right\} $
agents. For example, the consumers choosing $y$ ($w$) must evaluate
both features positively (resp., negatively), while this is less certain
for those choosing $x$ or $z$. If the population consists of only
linear types, then the third column of $M_{j}$ shows the shares of
the type $\left\{ 1,2\right\} $ consumers whose vectors of feature
values belong to the normal cones at points $x$, $y$, $z$, and
$w$.

\begin{figure}
\centering{}

\includegraphics{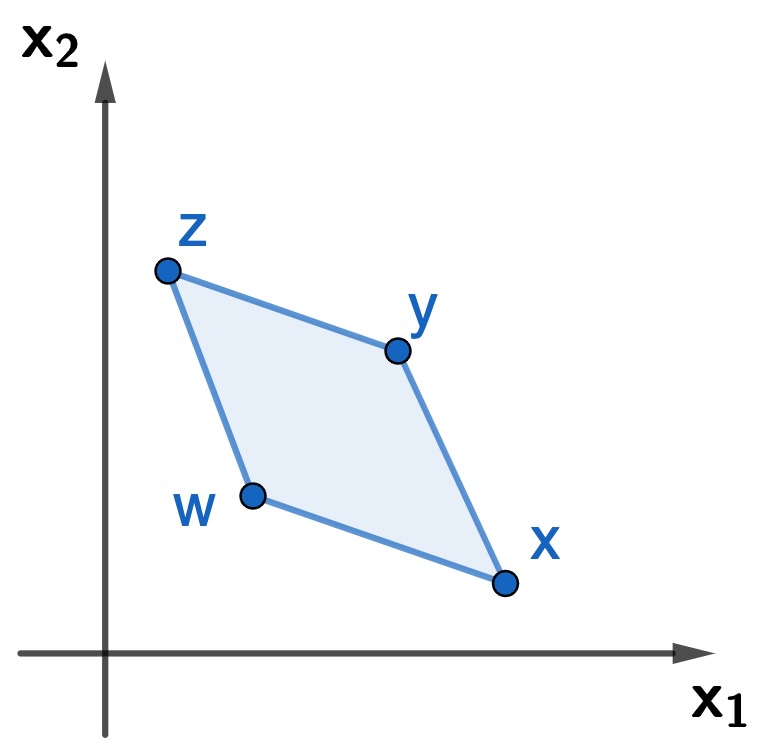}

\caption{Choice problem in Example \ref{example multiocc 1}}
\label{fig:rhombus}
\end{figure}

\subsection{Incomplete preferences}

Suppose that the set $X=\left\{ x_{1},\ldots,x_{n}\right\} $ of alternatives
is linearly ordered from best to worst according to a strict preference
relation $\succ$, i.e., $x_{1}\succ x_{2}\succ\dots\succ x_{n}$.
Agents' preferences are \emph{incomplete} in the sense that alternatives
that are sufficiently close in this ordering are treated as incomparable
or indifferent.\footnote{Incomplete preferences have been studied extensively in decision theory
and social choice, both as a primitive departure from classical rationality
(Aumann \cite{Aumann1962}, Dubra \textit{et al.} \cite{Dubra2004})
and as an outcome of bounded rationality or limited discrimination
(Masatlioglu \textit{et al.} \cite{Masatlioglu2012}, Manzini and
Mariotti \cite{manmar12}).} Formally, an agent of type $t_{l}$ can distinguish between two alternatives
only if they are at least $l$ positions apart in the ordering induced
by $\succ$. Let $\succ_{l}$ denote the preference relation of type
$t_{l}$. It follows that each $\succ_{l}$ is a coarsening of $\succ$:
$\succ_{1}=\succ$, $\succ_{l+1}\subset\succ_{l}$ for all $l$, and
$\succ_{n}$ is the empty relation. The set of types is given by $T=\left\{ t_{1},\dots t_{n}\right\} $.

Each type never chooses a dominated alternative and randomizes among
undominated ones. In this sense, $t_{1}$ represents the fully rational
type, whereas $t_{n}$ corresponds to the fully random type. Observe
that $t_{l}\land x_{k}$ if and only if $l\ge k$. Consequently, the
parameters populate an upper-triangular type-conditional choice matrix
$M$. This structure guarantees the existence of a unique matching,
and the parameters in $\pi$ are globally identifiable from the observed
choice probabilities $p=M\pi$ by Theorem \ref{prop:generic-ident}.
Moreover, if data from multiple choice occasions are available, the
analyst can also identify the entries of $M$ (the randomization parameters).
In particular, in the three-occasion setting discussed in Section
\ref{sec:Inference-from-multiple}, we have $v_{1}=v_{2}=v_{3}=n$
and $r=n$. Hence, both conditions of Theorem \ref{prop:ident-3occasions}
are satisfied, implying that the parameters in $\pi$, $M_{1},$ $M_{2},$
and $M_{3}$ are generically identifiable.
\begin{example}
\label{ex:incomplete-pref}Let $X=\left\{ x,y,z,w\right\} $ and $x\succ y\succ z\succ w$.
The induced preference relations for the different types are given
by
\[
\succ_{1}=\succ,\ \left\{ x\succ_{2}z,x\succ_{2}w,y\succ_{2}w\right\} ,\ \left\{ x\succ_{3}w\right\} ,\ \succ_{4}=\emptyset.
\]
The corresponding sets of undominated alternatives are therefore $\left\{ x\right\} $
for type $t_{1},$ $\left\{ x,y\right\} $ for type $t_{2}$, the
top three alternatives for type $t_{3}$, and all four alternatives
for type $t_{4}$. As a result, the type-conditional choice matrix
$M$ takes the form

\[
M = \bordermatrix{~ & t_1 & t_2 & t_3 & t_4 \cr
               x & \bullet & \bullet & \bullet & \bullet \cr
               y & 0 & \bullet & \bullet & \bullet \cr
               z & 0 & 0 & \bullet & \bullet \cr
               w & 0 & 0 & 0 & \bullet \cr}.
\]This matrix clearly admits a unique matching, implying that the type
distribution is globally identifiable. Moreover, if data from multiple
choice occasions are available, the choice matrix $M$ itself is also
generically identifiable.
\end{example}
The type-state specification of the framework is useful when the exact
nature of the agents\textquoteright{} randomization behavior is known.
To illustrate, suppose that agents resolve incompatibilities or indifferences
under the influence of the salience of the alternatives. Formally,
let each state $a\in A$ specify a linear salience ordering $\vartriangleright_{a}$
over the set $X$ of alternatives. An agent of type $t_{l}$ in state
$a$ chooses the maximizer of $\vartriangleright_{a}$ over the agent's
set $\left\{ x_{1},\dots,x_{l}\right\} $ of undominated alternatives.
Each state thus generates a type-conditional choice matrix $M^{a}$
with one-hot columns. The overall type-conditional choice matrix is
obtained by aggregation, $M=\sum_{a\in A}M^{a}$. According to Theorem
\ref{prop:ident-given-a}, the type distribution $\pi$ is generically
identifiable from $p=M\pi$ if the type is fully identified in at
least one state $a^{*}\in A$. 

One case when the type is fully identified given $a^{*}\in A$ arises
when the salience ordering $\vartriangleright_{a^{*}}$ is the \emph{reverse}
of the preference ordering $\succ$. In this case, each type $t_{l}$
selects $x_{l}$, i.e., $c_{t_{l}}(a^{*})=x_{l}$ for all $l$, and
the corresponding choice matrix $M^{a^{*}}$ is the identity matrix.
In fact, the presence of the salience ordering that is opposite to
$\succ$ is not only sufficient but also necessary for identification.
Indeed, such an ordering is the only one under which the worst alternative
$x_{n}$ can be matched to a type. Hence, it follows from Theorem
\ref{prop:gen-indep-cols-iff} that the parameters in $\pi$ are generically
identifiable if and only if there exists some $a^{*}\in A$ such that
$\vartriangleright_{a^{*}}$ is the reverse of $\succ$. 

As an illustration, consider Example \ref{ex:incomplete-pref} with
two states $A=\left\{ a,b\right\} $, where the salience ordering
$\vartriangleright_{a}$ coincides with the preference ordering $\succ$,
while $\vartriangleright_{b}$ is the reverse of $\succ$. Let $f(a)$
and $f(b)$ denote the probabilities of states $a$ and $b$, respectively.
The equation $f(a)M^{a}+f(b)M^{b}=M$ then takes the form

\[
f(a)\left[\begin{array}{cccc}
1 & 1 & 1 & 1\\
0 & 0 & 0 & 0\\
0 & 0 & 0 & 0\\
0 & 0 & 0 & 0
\end{array}\right]+f(b)\left[\begin{array}{cccc}
1 & 0 & 0 & 0\\
0 & 1 & 0 & 0\\
0 & 0 & 1 & 0\\
0 & 0 & 0 & 1
\end{array}\right]=\left[\begin{array}{cccc}
f(a)+f(b) & f(a) & f(a) & f(a)\\
0 & f(b) & 0 & 0\\
0 & 0 & f(b) & 0\\
0 & 0 & 0 & f(b)
\end{array}\right].
\]
For a generic choice of $f$, the resulting matrix $M$ is invertible,
implying that the distribution $\pi$ is globally identifiable.

\section{Concluding remarks}

We provide necessary and sufficient conditions under which aggregate
choice shares from a single set of alternatives guarantee identifiability
(generic or global) of the type distribution, as well as sufficient
conditions for generic identification of both the type distribution
and the type-level choice probabilities. A central message of our
analysis is that the informational content of aggregate choice data
depends less on its volume per se than on the structure of the qualitative
behavioral patterns it embeds. Strong identification results can be
obtained without imposing parametric structure on decision rules,
and without requiring rich menu variation or many repeated observations.
At the same time, our results clarify the limits of what can be learned
from aggregate data: when behavioral heterogeneity is insufficient,
even generic identification breaks down, regardless of the amount
of data available.

While the broad link between diversity in behvior and identifiability
is intuitive, a contribution of our analysis is to make the required
notion of \textquotedblleft qualitative behavioral heterogeneity\textquotedblright{}
precise. We do so by characterizing it combinatorially, through the
notion of ``matchings'' between types and alternatives, and by providing
an equivalent operationalization in algebraic terms, via interpretable
nullspace conditions on the relevant matrices.

By working with abstract behavioral types rather than more specific
primitives, the framework developed here is designed to apply to a
wide range of behavioral theories---such as limited consideration,
salience-based choice, status quo bias, or incomplete preferences---as
well as to more ad hoc behavioral hypotheses tailored to particular
contexts. Moreover, the qualitative nature of the characterizing properties
frees the researcher from the need for precise knowledge of type-level
choice probabilities, even at a merely parametric level.

We have focused on identification under extremely undemanding informational
assumptions. In many empirical settings, the analyst will have access
to additional information, such as observations from multiple menus,
covariates, or partial knowledge of the cognitive processes underlying
choice. An important direction for future research is to extend the
framework developed here to incorporate such additional sources of
information and to assess the extent to which they further strengthen
identification.

\clearpage{}

\appendix

\section{Proofs}

\subsection{Proof of Theorem \ref{prop:generic-ident}}\label{app:1}
\begin{proof}
Note that generic identifiability of $\pi$ here is equivalent to
the columns of $M$ being generically independent.

If the columns of $M$ are generically independent, then there exists
an $r\times r$ submatrix $M'$ of $M$ such that $\det(M')$ is not
identically 0. Note that $\det(M')$ is a multivariate polynomial
in the parameters $\left\{ M'(k,l):t_{l}\land x_{k}\right\} $. From
the Leibniz determinant formula 
\[
\det(M')=\sum_{\sigma\in S_{r}}\textup{sgn}(\sigma)\prod_{l=1}^{r}M'(\sigma(l),l),
\]
(where $S_{r}$ denotes the set of permutations on $\left\{ 1,...,r\right\} $)
it follows that there exists a permutation $\sigma^{*}$ of the rows
of $M'$ such that the monomial
\[
\prod_{l=1}^{r}M'(\sigma^{*}(l),l)
\]
 is not identically 0. This is possible only if all of the variables
in the monomial are not identically 0. Hence, $t_{l}\land x_{\sigma^{*}(l)}$
for all $l\in\left\{ 1,\dots,r\right\} $ and there exists a matching
$m$ defined as $m(t_{l})=x_{\sigma^{*}(l)}$. This proves one direction
of the assertion on generic identifiability in the statement.

Conversely, if there exists a matching $m$, then we can define a
permutation $\sigma^{*}$ for all $l\in\left\{ 1,\ldots,r\right\} $
as the unique permutation satisfying equation $m(t_{l})=x_{\sigma^{*}(l)}$.
Then, we have $t_{l}\land x_{\sigma^{*}(l)}$ for all $l\in\left\{ 1,\dots,r\right\} $.
Form an $r\times r$-submatrix $M'$ of $M$ using the rows corresponding
to the alternatives in the image of $m$. Choosing the point in the
parameter space defined by $M'(k,l)=1$ if $k=\sigma^{*}(l)$ and
$M'(k,l)=0$ otherwise, we have $\det(M')=1$. Hence, $\det(M')$
is not identically 0, which means that the $r$ columns of $M$ are
generically independent. This concludes the proof of the other direction.

When the matching $m$ is unique, for all $\sigma\in S_{r}$ with
$\sigma\neq\sigma^{\ast}$ we have $M'(\sigma(l),l)=0$ for some $l$
(otherwise we would have found another matching different from $m$),
whereas $\prod_{l=1}^{r}M'(\sigma^{*}(l),l)\neq0$ since $M'(\sigma^{*}(l),l)>0$
for all $l$. Equivalently, in the Leibniz expansion there is exactly
one term that is nonzero, and we conclude that $\det(M')\neq0$ for
all parameter values. This shows that the columns of $M$ are independent
for all parameter values, and $\pi$ is globally identified.

Finally, if there exists no matching, then for any $r\times r$-submatrix
$M'$ of $M$, all terms in the Leibniz expansion are zero, and therefore
$\det(M')=0$ for any such submatrix. Then the rank of $M$ is less
than $r$, and its columns are linearly dependent for all values of
the parameters, showing that $\pi$ is structurally non-identified.
\end{proof}

\subsection{Proof of Theorem \ref{thm:full-id-SNS}}\label{App:full-id-SNS}
\begin{proof}
As mentioned in the main text, the following adapts standard arguments
in the theory of SNS matrices (see e.g. \cite{brusha95}, Chapter
6). Note that the existence of a matching implies $n\geq r$. For
any $R\subseteq X$ with $|R|=r$, let $M_{R}$ denote the $r\times r$
submatrix of $M$ obtained by retaining only the rows corresponding
to elements of $R$. A matching $m:T\to R$ is identified with the
permutation $\sigma_{m}$ defined by $m(t_{i})=x_{\sigma_{m}(i)}$.
Condition (ii) of the theorem says that all matchings $m:T\to R$
have the same parity.

To prove sufficiency, assume there exists $R\subseteq X$ with $|R|=r$
such that: (i) there exists a matching $m:T\to R$; and (ii) every
matching $m':T\to R$ has the same parity as $m$. Consider any admissible
parameter vector $(\pi,M)$ (that is, $M$ respects $\wedge$ and
$\pi\in\Delta(T)$). We claim that $M_{R}$ is nonsingular for every
admissible choice of its positive entries consistent with $\wedge$.
Indeed, expand $\det(M_{R})$ by the Leibniz formula: 
\[
\det(M_{R})=\sum_{\sigma\in S_{r}}\mathrm{sgn}(\sigma)\prod_{l=1}^{r}M_{R}\bigl(\sigma(l),l\bigr).
\]
A term is nonzero if and only if the corresponding assignment corresponds
to a matching $T\to R$. By (i) at least one term is nonzero. By (ii),
all nonzero terms correspond to matchings of the same parity, hence
all have the same sign $\mathrm{sgn}(\sigma)$. Since every nonzero
term is a product of positive entries, all nonzero terms add with
the same sign, so the determinant cannot vanish: 
\[
\det(M_{R})\neq0.
\]
Therefore $M_{R}$ is invertible and $\pi$ is globally identifiable.

For the necessity part, we prove the contrapositive. Assume no subset
$R\subseteq X$ with $|R|=r$ satisfies conditions (i)--(ii). This
means that, for every $R\subseteq X$ with $|R|=r$, either: 
\begin{itemize}
\item (a) there is no matching $T\to R$; or 
\item (b) there exist at least two matchings $T\to R$ of opposite parity. 
\end{itemize}
We will construct an admissible $n\times r$ matrix $M^{\ast}$ (i.e.
one respecting $\wedge$ and column-stochasticity) with $\mathrm{rank}(M^{\ast})<r$.
This will show that $\pi$ is not globally identifiable. We proceed
in two steps:

\medskip{}
\emph{Step 1}: for each $R$, build an admissible $n\times r$ matrix
$M^{(R)}$ with $\det\!\bigl(M_{R}^{(R)}\bigr)=0$\textbf{.}

\noindent Fix $R$ with $|R|=r$.

Case (a): no matching $T\to R$. Then every permutation $\sigma$
uses at least one zero entry of $M_{R}$, so every Leibniz monomial
is zero and $\det(M_{R})\equiv0$ for all matrices with that zero
pattern. Hence we can pick any admissible (column-stochastic) $M^{(R)}$
respecting $\wedge$ and obtain $\det(M_{R}^{(R)})=0$.

Case (b): two matchings of opposite parity exist. Let $m_{+}$ and
$m_{-}$ be two such matchings. 

Fix a scalar $\lambda>0$. Define an $r\times r$ matrix $\widetilde{M}_{R}(\lambda)$
respecting the same zero pattern as $M_{R}$ by setting: 
\[
\widetilde{M}_{R}(\lambda)\bigl(\sigma_{m_{+}}(l),l\bigr)=\lambda\quad\text{for all }l,
\]
all other allowed (non-forced-zero) entries in $\widetilde{M}_{R}(\lambda)$
equal to $1$, and forced zeros equal to $0$. Then, in $\det\bigl(\widetilde{M}_{R}(\lambda)\bigr)$,
the monomial corresponding to $m_{+}$ equals $\mathrm{sgn}(\sigma_{m_{+}})\,\lambda^{r}$.
Every other nonzero monomial contains at most $r-1$ factors equal
to $\lambda$ (because it differs from $m_{+}$ in at least one position).
This implies that for $\lambda$ large enough the sign of $\det\bigl(\widetilde{M}_{R}(\lambda)\bigr)$
coincides with $\mathrm{sgn}(\sigma_{m_{+}})$. 

Analogously, define $\widehat{M}_{R}(\lambda)$ by placing $\lambda$
along the matching $m_{-}$ and $1$ on all other allowed entries;
then for $\lambda$ sufficiently large, the sign of $\det\bigl(\widehat{M}_{R}(\lambda)\bigr)$
is $\mathrm{sgn}(\sigma_{m_{-}})=-\mathrm{sgn}(\sigma_{m_{+}}).$ 

Next, we normalize columns to make them stochastic: let $D_{\widetilde{}}(\lambda)$
be the diagonal matrix whose $\ell$th diagonal entry is the sum of
column $\ell$ of $\widetilde{M}_{R}(\lambda)$, and set 
\[
M_{R}^{+}(\lambda)=\widetilde{M}_{R}(\lambda)\,D_{\widetilde{}}(\lambda)^{-1}.
\]
Similarly, define $M_{R}^{-}(\lambda)$ from $\widehat{M}_{R}(\lambda)$.
It holds that
\[
\det\!\bigl(M_{R}^{+}(\lambda)\bigr)=\det\!\bigl(\widetilde{M}_{R}(\lambda)\bigr)\cdot\det\!\bigl(D_{\widetilde{}}(\lambda)^{-1}\bigr),
\]
and analogously for $\det\!\bigl(M_{R}^{+}(\lambda)\bigr)$. Since
$\det(D_{\widetilde{}}(\lambda)^{-1})>0$, the sign is preserved,
Thus, for $\lambda$ large enough, $\det(M_{R}^{+}(\lambda))>0$ and
$\det(M_{R}^{-}(\lambda))<0$.

Finally, we extend $M_{R}^{+}(\lambda)$ and $M_{R}^{-}(\lambda)$
to full $n\times r$ column-stochastic matrices $M^{+}$ and $M^{-}$
respecting $\wedge$ by choosing arbitrary allowed positive entries
on rows $X\setminus R$ and then renormalizing each column (again
preserving the sign/nonzeroness of the minor on $R$ up to a positive
factor). In particular, we can ensure: 
\[
\det\!\bigl(M_{R}^{+}\bigr)>0\qquad\text{and}\qquad\det\!\bigl(M_{R}^{-}\bigr)<0.
\]
By the continuity of the determinant along convex combinations, there
exists $\alpha\in(0,1)$ such that the matrix $M^{\left(R\right)}$
defined by
\[
M^{(R)}:=\alpha M^{+}+(1-\alpha)M^{-}
\]
satisfies $\det\!\bigl(M_{R}^{(R)}\bigr)=0$. Since the set of matrices
that respect the zero pattern and column-stochasticity is clearly
convex, $M^{(R)}$ is admissible.

This completes Step~1.

\medskip{}
\emph{Step 2}: by combining across the $R$, construct the desired
admissible $M^{\ast}$ with all $r\times r$ minors equal to $0$.

\noindent Let $\mathcal{R}=\{R\subseteq X:\ |R|=r\}$. Define 
\[
M^{\ast}\;:=\;\frac{1}{|\mathcal{R}|}\sum_{R\in\mathcal{R}}M^{(R)}.
\]
Being a convex combination of admissible matrices, $M^{\ast}$ is
admissible.

Fix any $R_{0}\in\mathcal{R}$. Consider the polynomial map that takes
$M$ to $\det(M_{R_{0}})$. This is multilinear in the rows indexed
by $R_{0}$ and continuous. By construction, among the summands in
$M^{\ast}$, the matrix $M^{(R_{0})}$ has $\det\bigl(M_{R_{0}}^{(R_{0})}\bigr)=0$.
Moreover, since in Step~1 we are free to choose each $M^{(R)}$ and
we only require $\det(M_{R}^{(R)})=0$ for its own $R$, we may (and
do) choose each $M^{(R)}$ so that all $r\times r$ minors vanish.\footnote{One way to do this is: in Step~1, after obtaining an admissible $M^{(R)}$
with $\det(M_{R}^{(R)})=0$, replace it by a convex combination with
a fixed admissible matrix of rank $<r$ respecting the same zero pattern
(which exists under our standing assumption), so that all minors vanish
while preserving admissibility.} Hence every $r\times r$ minor of $M^{\ast}$ is $0$, implying 
\[
\mathrm{rank}(M^{\ast})<r.
\]

\medskip{}
This means that the columns of $M^{*}$ are linearly dependent and
therefore $\pi$ is not globally identified.
\end{proof}

\subsection{Proof of Theorem \ref{prop:ident-given-a}}\label{app:2}
\begin{proof}
We will show that the columns of $M$ are generically independent.
Suppose that $a^{*}\in A$ separates types. Then there exists a matching
$m:T\to X$ such that, for all $t\in T$, $m(t)=c_{t}\left(a^{*}\right)$.
This implies that $n\geq r$. Form an $r\times r$-submatrix $M'$
of $M$ using the rows $\left\{ k:x_{k}=m(t),\;t\in T\right\} $.
The determinant of $M'$ is a polynomial in the parameters $\left\{ f(s)\right\} _{s\in A}$.
Consider the point in the parameter space where $f(s)=1$ if $s=a^{*}$
and $f(s)=0$ otherwise. At this point, we have that $M'$ is a permutation
matrix, which implies $det(M')\neq0$. Hence, the zero set of $det(M')$
is a proper subvariety and not the whole parameter space. We have
just shown that there exists an $r\times r$-minor of $M$ such that
its zero set is a proper subvariety. This implies that the columns
of $M$ are generically independent. 
\end{proof}

\subsection{Proof of Theorem \ref{prop:gen-indep-cols-iff}}\label{app:3}
\begin{proof}
Fix $R\subseteq X$ with $|R|=r$. For $a\in A$, it holds that 
\[
M_{R}=\sum_{a\in A}f(a)\,M_{R}^{a},
\]
where $M_{R}^{a}$ is the restriction of $M^{a}$ to the rows in $R$,
i.e. $M_{R}^{a}=\left(M^{a}\right)_{R}$. By the Leibniz formula,
and identifying each permutation in $S_{r}$ with the corresponding
matching, 
\[
\det(M_{R})=\sum_{m}\delta(m)\prod_{t\in T}M_{R}\bigl(m(t),t\bigr).
\]
where the summation is over all matchings $m:T\to R$ and $\delta\left(m\right)\in\left\{ -1,1\right\} $
is determined by the parity of the permutation associated with $m$.
For any matching $m:T\to R$ and any $t\in T$, 
\begin{align*}
M_{R}\bigl(m(t),t\bigr) & =\sum_{a\in A:c_{t}(a)=m(t)}f(a)=\sum_{\gamma:(m,\gamma)\in\mathcal{S}_{R}}f\left(\gamma\left(t\right)\right).
\end{align*}
Expanding the product of sums and observing that each choice of admissible
states $\left(\gamma\left(t\right)\right)_{t\in T}$ uniquely defines
a state-matching $\left(m,\gamma\right)\in{\cal S}_{R}$, we obtain
\begin{equation}
\det(M_{R})=\sum_{(m,\gamma)\in\mathcal{S}_{R}}\delta(m)\prod_{t\in T}f\bigl(\gamma(t)\bigr),\label{eq:det-sum-over-state-matchings-R}
\end{equation}
i.e.\ the determinant is a polynomial in $\{f(a)\}_{a\in A}$ whose
monomials are indexed precisely by state-matchings with image $R$.

\emph{Sufficiency}. Assume there exists $R\subseteq X$, $|R|=r$,
satisfying (i)--(ii). By (i), $\mathcal{S}_{R}\neq\varnothing$,
hence \eqref{eq:det-sum-over-state-matchings-R} contains at least
one monomial with nonzero coefficient. Condition (ii) rules out complete
cancellation of all monomials: indeed, cancellation of \eqref{eq:det-sum-over-state-matchings-R}
within every $\Gamma$-class would allow one to partition $\mathcal{S}_{R}$
into disjoint pairs $\big((m,\gamma),(m',\gamma')\big)$ with $\Gamma(m,\gamma)=\Gamma(m',\gamma')$
and $\delta(m')=-\delta(m)$ for every pair, contradicting (ii). Therefore,
$\det(M_{R})$ is not the zero polynomial. Hence $M$ has generically
full column rank $r$, and $\pi$ is generically identifiable from
$p=M\pi$.

\emph{Necessity.} Suppose $\pi$ is generically identifiable. Then
$M$ has generically full column rank $r$, so there exists at least
one $r\times r$ minor $M_{R}$ whose determinant is not identically
zero as a polynomial in $\{f(a)\}_{a\in A}$. Fix such an $R$. If
$\mathcal{S}_{R}=\varnothing$, then every product in the Leibniz
formula is identically zero, so $\det(M_{R})\equiv0$, a contradiction.
Thus (i) holds.

If (ii) fails for this $R$, then there exists a partition of $\mathcal{S}_{R}$
into disjoint pairs $\big((m,\gamma),(m',\gamma')\big)$ such that
$\Gamma(m,\gamma)=\Gamma(m',\gamma')$ and $m'$ differs from $m$
by an odd number of swaps, i.e.\ $\delta(m')=-\delta(m)$, for every
pair. For each pair, $\Gamma(m,\gamma)=\Gamma(m',\gamma')$ implies
\[
\prod_{t\in T}f\bigl(\gamma(t)\bigr)=\prod_{t\in T}f\bigl(\gamma'(t)\bigr),
\]
while $\delta(m')=-\delta(m)$ implies the corresponding contributions
in \eqref{eq:det-sum-over-state-matchings-R} cancel. Hence all terms
cancel and $\det(M_{R})\equiv0$, again contradicting the choice of
$R$. Therefore (ii) must hold for this $R$.

\emph{Structural non-identifiability}. Suppose that no $R\subseteq X$
with $|R|=r$ satisfies (i)--(ii). Then, for every $R$, the polynomial
$\det(M_{R})$ is identically zero, hence every $r\times r$ minor
of $M$ vanishes for every parameter value. Thus, $\operatorname{rank}(M)<r$
for all admissible parameters, and $\pi$ is structurally non-identifiable.
\end{proof}

\subsection{Proof of Theorem \ref{cor:global_id_typestate}}\label{App:Cor_global_id_typestate}
\begin{proof}
Fix $R\subseteq X$ with $|R|=r$. From the determinant expansion
established in the proof of Theorem \ref{prop:gen-indep-cols-iff},
with $M_{R}$ defined as there, we have 
\begin{equation}
\det(M_{R})=\sum_{(m,\gamma)\in\mathcal{S}_{R}}\delta(m)\prod_{t\in T}f(\gamma(t)),\label{eq:det-expansion}
\end{equation}
where $\delta(m)\in\{-1,1\}$ denotes the parity of the matching $m$.

\medskip{}
\emph{Sufficiency.} If condition (ii) holds, then $\delta(m)$ is
constant over all $(m,\gamma)\in\mathcal{S}_{R}$. Hence every nonzero
monomial in $\det(M_{R})$ has the same sign. The condition $\mathcal{S}_{R}\neq\varnothing$
ensures that at least one such monomial exists, and therefore $\det(M_{R})\neq0$
for every state distribution $f$ in the interior of the simplex.
It follows that $M_{R}$ is nonsingular for all admissible parameter
values, and $\pi$ is globally identifiable.

\medskip{}
\noindent\emph{Necessity. }Fix $R\subseteq X$ with $|R|=r$. If $\mathcal{S}_{R}=\varnothing$,
then the determinant expansion contains no monomials and hence it
is identically zero, in which case $M_{R}$ is singular for every
admissible state distribution. Therefore global identifiability requires
$\mathcal{S}_{R}\neq\varnothing$. 

\noindent Next, suppose condition (i) holds but condition (ii) fails.
Then there exist two state-matchings $(m,\gamma),(m',\gamma')\in\mathcal{S}_{R}$
with opposite parity, i.e.\ $\delta(m)=-\delta(m')$, and in the
expansion of $\det(M_{R})$ the monomials 
\[
\prod_{t\in T}f(\gamma(t))\qquad\text{and}\qquad\prod_{t\in T}f(\gamma'(t))
\]
appear with opposite signs.

\noindent Let 
\[
A^{*}:=\{\gamma(t):t\in T\}\cup\{\gamma'(t):t\in T\}
\]
be the set of states used by these two state-matchings. Fix a small
$\varepsilon>0$ and construct a strictly positive $f^{+}\in\Delta^ {}(A)$
as follows: assign probability mass $\varepsilon$ to each state in
$A\setminus A^{*}$, assign probability mass $\varepsilon$ to each
state in $A^{*}\setminus\{\gamma(t):t\in T\})$, and distribute the
remaining probability mass uniformly over the set $\{\gamma(t):t\in T\}$.
In this way all components of $f^{+}$ are strictly positive. For
$\varepsilon$ small enough the contribution of the monomial indexed
by $(m,\gamma)$ dominates the sum of all remaining monomials in \eqref{eq:det-expansion},
so that we obtain $\det(M_{R})>0$ at the parameter values $f^{+}$.

\noindent Analogously, define a strictly positive $f^{-}\in\Delta(A)$
by making the states in $\{\gamma'(t):t\in T\}$ receive the bulk
of the probability mass and all other states probability of order
$\varepsilon$. For $\varepsilon$ small enough this yields $\det(M_{R})<0$
at the parameter values $f^{-}$.

\noindent Finally, observe that the set of strictly positive points
in $\Delta(A)$ is convex, and by the continuity of the determinant
map there exists a combination $f^{\ast}\in\Delta^ {}(A)$ of $f^{+}$
and $f^{-}$ such that $\det(M_{R})=0$ at the parameter values $f^{\ast}$.
We have shown that $M_{R}$ is singular for some strictly positive
$f$, so that $\pi$ is not globally identifiable on the interior
of the simplex.
\end{proof}

\subsection{Proof of Corollary \ref{Th: global_typestate_sufficient}}
\begin{proof}
Fix the separating state $a^{*}$. By the definition of a separating
state, there exists a matching $m_{a^{*}}:T\to X$ with $m_{a^{*}}(t)=c_{t}(a^{*})$.
Write $R=\text{Im}(m_{a^{*}})$. If there exists a possible reassignment
$\varphi$ that is odd, then for each $t\in T$ choose a state $\gamma(t)\in A$
such that 
\[
c_{t}(\gamma(t))=c_{\varphi(t)}(a^{*}).
\]
Define $m:T\to R$ by $m(t):=c_{\varphi(t)}(a^{*})$. Since $\varphi$
is bijective, $\text{Im}(m)=R$, and $(m,\gamma)$ is a state-matching
for which $m$ differs from $m_{a}$ by an odd permutation. Hence,
by Theorem \ref{cor:global_id_typestate}, global identifiability
fails.

Conversely, suppose global identifiability fails. We show that there
exists a possible reassignment $\varphi$ that is odd. By Theorem
\ref{cor:global_id_typestate} there exist two state-matchings $(m,\gamma)$
and $(m_{a^{*}},a^{*})$ with $\text{Im}(m)=R$ such that $m$ differs
from $m_{a^{*}}$ by an odd number of swaps. Define $\varphi:=m_{a^{*}}^{-1}\circ m$.
Then $\varphi$ is an odd permutation of $T$. Moreover, for each
$t\in T$ we have 
\[
c_{t}(\gamma(t))=m(t)=c_{\varphi(t)}(a^{*}),
\]
so $t\leadsto\varphi(t)$ and $\varphi$ is the required possible
reassignment.
\end{proof}

\subsection{Proof of Theorem \ref{nullspace2by2}}\label{app:4}

\label{proof_nullspace2by2} 
\begin{proof}
$(i)\Rightarrow(ii):$ Suppose statement (ii) of the theorem fails,
i.e., ${\bf x}\in\operatorname{Null}(M^{a})\cap\operatorname{Null}(M^{b})$
and ${\bf x}\neq0$. Then,
\[
M{\bf x}=(f(a)M^{a}+f(b)M^{b}){\bf x}=f(a)M^{a}{\bf x}+f(b)M^{b}{\bf x}={\bf 0},
\]
irrespective of the values of $f(a)$ and $f(b)$. Therefore, for
all $f(a)$ and $f(b)$, the matrix $M=f(a)M^{a}+f(b)M^{b}$ is not
invertible. Hence, $M$ is not generically invertible, and the parameters
in $\pi$ are not generically identifiable.

$(ii)\Rightarrow(iii):$ Suppose statement (iii) of the theorem fails,
i.e., both $c_{t_{1}}(a)=c_{t_{2}}(a)$ and $c_{t_{1}}(b)=c_{t_{2}}(b)$.
This implies that the columns of $M^{a}$ corresponding to $t_{1}$
and $t_{2}$ are the same vectors. Analogously, for $M^{b}$. Since
these columns are one-hot vectors, it follows that 
\[
M^{a},M^{b}\in\left\{ \begin{bmatrix}1 & 1\\
0 & 0
\end{bmatrix},\begin{bmatrix}0 & 0\\
1 & 1
\end{bmatrix}\right\} 
\]
Note that irrespective of which of these two forms either of the two
matrices $M^{a}$ and $M^{b}$ take, it holds that 
\[
\operatorname{Null}(M^{a})=\operatorname{Null}(M^{b})=\operatorname{span}\left\{ \begin{bmatrix}\;\;\,1\\
-1
\end{bmatrix}\right\} 
\]
But, then $\operatorname{Null}(M^{a})\cap\operatorname{Null}(M^{b})\neq\{{\bf 0}\}$.
This concludes this part of the proof.

$(iii)\Rightarrow(i):$ Note that, WLOG we can assume that $c_{t_{1}}(a)\neq c_{t_{2}}(a)$.
So, the columns of $M^{a}$ corresponding to $t_{1}$ and $t_{2}$
are different vectors. Since these columns are one-hot vectors, this
implies that 
\[
M^{a}\in\left\{ \begin{bmatrix}1 & 0\\
0 & 1
\end{bmatrix},\begin{bmatrix}0 & 1\\
1 & 0
\end{bmatrix}\right\} 
\]
Observe that, any such $M^{a}$ is invertible. As such, by Theorem
\ref{invertiblegeneric}, it follows that $M=f(a)M^{a}+f(b)M^{b}$
is generically invertible. This concludes the proof. 
\end{proof}

\subsection{Proof of Theorem \ref{nullspace3by3}}\label{app:5}

\label{proof_nullspace3by3} 
\begin{proof}
$(i)\Rightarrow(ii):$ Suppose condition (ii) fails. First, suppose
that $\operatorname{Null}(M^{a})\cap\operatorname{Null}(M^{b})\neq\{{\bf 0}\}$.
Then, fix any non-trivial ${\bf z}$, satisfying ${\bf z}\in\operatorname{Null}(M^{a})$
and ${\bf z}\in\operatorname{Null}(M^{b})$ and observe that 
\[
M{\bf z}=(f(a)M^{a}+f(b)M^{b}){\bf z}=f(a)M^{a}{\bf z}+f(b)M^{b}{\bf z}={\bf 0}
\]
Since ${\bf z}$ is non-trivial, this implies that $\operatorname{Null}(M)\neq\{{\bf 0}\}$.
By the Invertible Matrix Theorem, it follows that $M$ is not invertible.

Next, suppose that $\operatorname{Null}(M^{a^{T}})\cap\operatorname{Null}(M^{b^{T}})\neq\{{\bf 0}\}$.
Then, fix any non-trivial ${\bf w}$, satisfying ${\bf w}\in\operatorname{Null}(M^{a^{T}})$
and ${\bf w}\in\operatorname{Null}(M^{b^{T}})$ and observe that 
\[
M{\bf w}=(f(a)M^{a^{T}}+f(b)M^{b^{T}}){\bf w}=f(a)M^{a^{T}}{\bf w}+f(b)M^{b^{T}}{\bf w}={\bf 0}
\]
Since ${\bf w}$ is non-trivial, this implies that $\operatorname{Null}(M^{T})\neq\{{\bf 0}\}$.
By the Invertible Matrix Theorem, it follows that $M$ is not invertible.
This concludes this part of the proof.

$(ii)\Rightarrow(iii):$ Suppose statement (iv) of the theorem fails,
i.e., there exist $t,t'\in T$ such that, for all $a\in A$, it holds
that $c_{t}(a)=c_{t'}(a)$, or, there exists an alternative $w\in X$
such that, for all $a\in A$ and all $t\in T$, it holds that $c_{t}(a)\neq w$.

First, consider the case of there existing $t,t'\in T$ such that,
for all $a\in A$, it holds that $c_{t}(a)=c_{t'}(a)$. WLOG, we can
assume that $t$ and $t'$ correspond to the first two columns of
the matrices. Then, this implies that, for all $a\in A$, the first
two columns of $M^{a}$ are the same vectors. In other words, for
all $a\in A$, it holds that 
\[
\begin{bmatrix}\;\;\,1\\
-1\\
\;\;\,0
\end{bmatrix}\in\operatorname{Null}(M^{a})
\]
But then, $\operatorname{Null}(M^{a})\cap\operatorname{Null}(M^{b})\neq\{{\bf 0}\}$.

Next, consider the case of there existing an alternative $w\in X$
such that, for all $a\in A$ and all $t\in T$, it holds that $c_{t}(a)\neq w$.
WLOG, we can assume that this $w$ corresponds to the last column
of the matrices. Then, this implies that, for all $a\in A$, the last
row of $M^{a}$ is the zero vector. In other words, for all $a\in A$,
it holds that 
\[
\begin{bmatrix}0\\
0\\
1
\end{bmatrix}\in\operatorname{Null}(M^{{a}^{T}})
\]
But then, $\operatorname{Null}(M^{{a}^{T}})\cap\operatorname{Null}(M^{{b}^{T}})\neq\{{\bf 0}\}$.
This concludes this part of the proof.

$(iii)\Rightarrow(i):$ For all $t,t'\in T$ with $t\neq t'$, there
exists an $a\in A$ such that $c_{t}(a)\neq c_{t'}(a)$, and, for
all $w\in X$, there exist $a\in A$ and $t\in T$ such that $c_{t}(a)=w$.
Note that this implies that not all terms in the Leibniz formula for
the determinant of $M$ are zero. Since the columns of $M^{a}$ and
$M^{b}$ are one-hot vectors, this implies that at least one such
term is unique in its form $f(a)^{\alpha}f(b)^{3-\alpha}$, for some
$\alpha\in\{1,2,3\}$. But, then $M$ is generically invertible. This
concludes the proof. 
\end{proof}

\subsection{Proof of Theorem \ref{invertiblegeneric}}\label{app:6}

\label{proof_invertiblegeneric} 
\begin{proof}
Consider an $r\times r$ matrix $N=\sum_{b\in B}f(b)M^{b}$ which
is generically invertible. By the Invertible Matrix Theorem, it follows
that $\det(N)\neq0$, for all $\{f(b)\}_{b\in B}$ but some Lebesgue
measure zero set. Fix an arbitrary $r\times r$ matrix $M^{a}$ and
define $M(\alpha)=\alpha M^{a}+(1-\alpha)N$ for all $\alpha\in\left[0,1\right]$.
Observe that $\det(M(\alpha))$ is a polynomial of degree at most
$r$. Since $\det(M(0))=\det(N)\neq0$, for all $\{f(b)\}_{b\in B}$
but some Lebesgue measure zero set, it follows that $\det(M(\alpha))$
is not the zero polynomial for all such $\{f(b)\}_{b\in B}$. This
implies that the polynomial vanishes for at most $r$ values of $\alpha$.
Since $r$ is finite, this implies that $\det(M(\alpha))\neq0$, for
all $\{f(a)\}_{a\in A}$ but some Lebesgue measure zero set. By the
Invertible Matrix Theorem, it follows that $M=\alpha M^{a}+(1-\alpha)N$
is generically invertible. 
\end{proof}

\subsection{Proof of Corollary \ref{invertiblethengenericinvertible}}\label{app:7}

\label{proof_invertiblethengenericinvertible} 
\begin{proof}
Fix a probability measure $\overline{f}(a)$ on $A$ such that $\overline{M}=\sum_{a\in A}\overline{f}(a)M^{a}$
is invertible. Now, observe that ${M}=\sum_{a\in A}{f(a)}M^{a}$,
for any probability measure $f$ on $A$, can be written as a convex
combination of $\overline{M}$ and a matrix $N$ which itself is some
convex combination of the $M^{a}$'s. In short, $M=\alpha\overline{M}+(1-\alpha)N$,
for some $\alpha\in(0,1)$ and $N$ given by $N=\sum_{a\in A}\underline{f}(a)M^{a}$,
for some probability measure $\underline{f}$ on $A$. Since $\overline{M}$
is invertible, Theorem \ref{invertiblegeneric} implies that $M=\alpha\overline{M}+(1-\alpha)N$
is generically invertible. Note that this is true for any such $N$.
Hence, ${M}=\sum{f(a)}M^{a}$ is generically invertible. 
\end{proof}

\subsection{Proof of Theorem \ref{nullspaces_general}}\label{app:8}

\label{proof_nullspaces_general} 
\begin{proof}
$(i)\Rightarrow(ii):$ Let $A=\{a,b\}$ and suppose that $M=f(a)M^{a}+f(b)M^{b}$
is generically invertible. By the Invertible Matrix Theorem, it then
follows that, for all $\{f(a)\}_{a\in A}$ but some Lebesgue-measure
zero set, it holds that $\operatorname{Null}(f(a)M^{a}+f(b)M^{b})=\{{\bf 0}\}$.
Now, towards a proof by contradiction, suppose that there exists some
${\bf z}$ with ${\bf z}\neq{\bf x}+{\bf y}$ such that for some non-trivial
${\bf x}\in\operatorname{Null}(M^{a})$ and some non-trivial ${\bf y}\in\operatorname{Null}(M^{b})$,
it holds that 
\[
(f(a)M^{a}+f(b)M^{b}){\bf z}=M^{a}{\bf y}+M^{b}{\bf x}
\]
Note that then there also exists some non-trivial ${\bf x}\in\operatorname{Null}(M^{a})$
and some non-trivial ${\bf y}\in\operatorname{Null}(M^{b})$ such
that 
\[
(f(a)M^{a}+f(b)M^{b}){\bf z}=f(a)M^{a}{\bf y}+f(b)M^{b}{\bf x}
\]
Now, that this implies that 
\[
(f(a)M^{a}+f(b)M^{b}){\bf w}={\bf 0}
\]
for ${\bf w}$ defined as ${\bf w}={\bf z}-{\bf x}-{\bf y}$. To see
this, observe that 
\begin{align*}
 &  & {\bf 0} & =(f(a)M^{a}+f(b)M^{b}){\bf w}\\
\Leftrightarrow &  & {\bf 0} & =(f(a)M^{a}+f(b)M^{b})({\bf z}-{\bf x}-{\bf y})\\
\Leftrightarrow &  & (f(a)M^{a}+f(b)M^{b}){\bf z} & =(f(a)M^{a}+f(b)M^{b})({\bf x}+{\bf y})\\
\Leftrightarrow &  & (f(a)M^{a}+f(b)M^{b}){\bf z} & =f(a)M^{a}{\bf x}+f(a)M^{a}{\bf y}+f(b)M^{b}{\bf x}+f(b)M^{b}{\bf y}\\
\Leftrightarrow &  & (f(a)M^{a}+f(b)M^{b}){\bf z} & =f(a)M^{a}{\bf y}+f(b)M^{b}{\bf x}
\end{align*}
where we used that ${\bf x}\in\operatorname{Null}(M^{a})$ implies
that $f(a)M^{a}{\bf x}={\bf 0}$ and that ${\bf y}\in\operatorname{Null}(M^{b})$
implies that $f(b)M^{b}{\bf y}={\bf 0}$. Since ${\bf z}\neq{\bf x}+{\bf y}$,
it follows that ${\bf w}={\bf z}-{\bf x}-{\bf y}\neq{\bf 0}$. As
such, $(f(a)M^{a}+f(b)M^{b}){\bf w}={\bf 0}$ implies that $\operatorname{Null}(f(a)M^{a}+f(b)M^{b})\neq\{{\bf 0}\}$.
By the Invertible Matrix Theorem, it follows that $M=(f(a)M^{a}+f(b)M^{b})\in\mathcal{M}_{r\times r}(\mathbb{R})$
is not generically invertible. Hence, we have arrived at our desired
contradiction. This concludes this part of the proof.

$(ii)\Rightarrow(iii):$ Suppose that condition (iii) in the statement
of the theorem does not hold. This implies that there exists $t,t'\in T$
such that WLOG $c_{t}(a)=c_{t'}(a)$, while the split between $c_{t}(b)$
and $c_{t'}(b)$ is typical of $T\setminus\{t,t'\}$.

First, consider the case where $c_{t}(b)=c_{t'}(b)$. That is, Column
$t$ and $t'$ of matrix $M^{b}$ coincide. Since $c_{t}(a)=c_{t'}(a)$,
also Column $t$ and $t'$ of matrix $M^{a}$ coincide. It follows
that there exists a non-trivial ${\bf x}\in\mathbb{R}^{r}$ such that
${\bf x}\in\operatorname{Null}(M^{a})$ and ${\bf x}\in\operatorname{Null}(M^{b})$.
That is, $M^{a}{\bf x}={\bf 0}$ and $M^{b}{\bf y}={\bf 0}$. But
then, fixing ${\bf z}={\bf 0}$ and ${\bf y}={\bf x}$, we have some
non-trivial a ${\bf x}\in\operatorname{Null}(M^{a})$, some non-trivial
${\bf y}\in\operatorname{Null}(M^{b})$, and some ${\bf z}\in\operatorname{span}\{{\bf x},{\bf y}\}^{\perp}$
such that 
\[
(f(a)M^{a}+f(b)M^{b}){\bf z}={\bf 0}=f(a)M^{a}{\bf y}+f(b)M^{b}{\bf x}
\]

Next, consider the case where $c_{t}(b)\neq c_{t'}(b)$. Definition
\ref{typical} then implies that the set $T\setminus\{t,t'\}$ is
non-empty. But, then the split between $c_{t}(b)$ and $c_{t'}(b)$
being typical of $T\setminus\{t,t'\}$ implies that 
\[
({\bf e}^{c_{t}(b)}-{\bf e}^{c_{t'}(b)})\in\operatorname{span}\{M_{*j}:\text{column }j\text{ corresponds to some type }s\in T\setminus\{t,t'\}\}
\]
where ${\bf e}^{w}$, for any $w\in\{x,y\}$, is defined such that,
for all $i=1,\dots,r$ 
\[
{\bf e}_{i}^{w}=\begin{cases}
1, & \text{if Row \ensuremath{i} corresponds to alternative \ensuremath{w}}\\
0, & \text{otherwise}
\end{cases}
\]
Defining ${\bf x}\in\mathbb{R}^{r}$ via 
\[
{\bf x}_{i}=\begin{cases}
\;\;\,1, & \text{if Row \ensuremath{i} corresponds to alternative \ensuremath{c_{t}(b)}}\\
-1, & \text{if Row \ensuremath{i} corresponds to alternative \ensuremath{c_{t'}(b)}}\\
\;\;\,0, & \text{otherwise}
\end{cases}
\]
this implies that 
\[
M^{b}{\bf x}\in\operatorname{span}\{M_{*j}:\text{column }j\text{ corresponds to some type }s\in T\setminus\{t,t'\}\}
\]
In other words, there exists ${\bf w}\in\operatorname{span}\{{\bf x}\}^{\perp}$
such that 
\begin{equation}
(f(a)M^{a}+f(b)M^{b}){\bf w}=f(b)M^{b}{\bf x}
\label{xx}
\end{equation}
By Theorem \ref{invertiblegeneric}, there also exist $s,s'\in T$
such that $c_{s}(b)=c_{s'}(b)$, while the split between $c_{s}(a)$
and $c_{s'}(a)$ is typical of $T\setminus\{s,s'\}$. As such, analogously
to above, we can define ${\bf y}\in\mathbb{R}^{r}$ via 
\[
{\bf y}_{i}=\begin{cases}
\;\;\,1, & \text{if Row \ensuremath{i} corresponds to alternative \ensuremath{c_{s}(a)}}\\
-1, & \text{if Row \ensuremath{i} corresponds to alternative \ensuremath{c_{s'}(a)}}\\
\;\;\,0, & \text{otherwise}
\end{cases}
\]
such that there exists ${\bf v}\in\operatorname{span}\{{\bf y}\}^{\perp}$
satisfying 
\begin{equation}
(f(a)M^{a}+f(b)M^{b}){\bf v}=f(a)M^{a}{\bf y}
\label{yy}
\end{equation}
Combing Equations \eqref{xx} and \eqref{yy}, we find that there
exists some non-trivial ${\bf x}\in\operatorname{Null}(M^{a})$, some
non-trivial ${\bf y}\in\operatorname{Null}(M^{b})$, and some ${\bf z}\in\operatorname{span}\{{\bf x},{\bf y}\}^{\perp}$
such that 
\[
(f(a)M^{a}+f(b)M^{b}){\bf z}=f(a)M^{a}{\bf y}+f(b)M^{b}{\bf x}
\]
Note that ${\bf x}\in\operatorname{Null}(M^{a})$ implies that also
$(f(b){\bf x})\in\operatorname{Null}(M^{a})$ and that ${\bf y}\in\operatorname{Null}(M^{b})$
implies that also $(f(a){\bf y})\in\operatorname{Null}(M^{b})$. Hence,
this implies that there also exists some non-trivial ${\bf x}\in\operatorname{Null}(M^{a})$,
some non-trivial ${\bf y}\in\operatorname{Null}(M^{b})$, and some
${\bf z}\in\operatorname{span}\{{\bf x},{\bf y}\}^{\perp}$ such that
\[
(f(a)M^{a}+f(b)M^{b}){\bf z}=M^{a}{\bf y}+M^{b}{\bf x}
\]

$(iii)\Rightarrow(i):$ Fix some $a,b\in A$ with $a\neq b$ and suppose
that, for all $t,t'\in T$ where $c_{t}(a)=c_{t'}(a)$, the split
between $c_{t}(b)$ and $c_{t'}(b)$ is not typical of $(T\setminus\{t,t'\})$.
For any such $t$ and $t'$, define ${\bf x}\in\mathbb{R}^{r}$ via
\[
{\bf x}_{i}=\begin{cases}
\;\;\,1, & \text{if the \ensuremath{i}'th column of \ensuremath{M} corresponds to type \ensuremath{t}}\\
-1, & \text{if the \ensuremath{i}'th column of \ensuremath{M} corresponds to type \ensuremath{t'}}\\
\;\;\,0, & \text{otherwise}
\end{cases}
\]
Then, ${\bf x}\in\operatorname{Null}(M^{a})$. Now, the split between
$c_{t}(b)$ and $c_{t'}(b)$ not being typical of $(T\setminus\{t,t'\})$
implies that, for any such non-trivial ${\bf x}$ and all ${\bf z}\in\operatorname{span}\{{\bf x}\}^{\perp}$,
it holds that 
\[
M{\bf z}\neq M^{b}{\bf x}
\]
Since ${\bf x}\in\operatorname{Null}(M^{a})$ and, thus, $M^{a}{\bf x}={\bf 0}$,
this implies that, for $f(b)\neq0$, it holds that 
\[
M{\bf z}\neq(f(a)M^{a}+f(b)M^{b}){\bf x}
\]
But, since any ${\bf w}\in\operatorname{Null}(M)$ can be written
as ${\bf w}={\bf z}-{\bf x}$, for some ${\bf x}\in\operatorname{Null}(M^{a})$
and some ${\bf z}\in\operatorname{span}\{{\bf x}\}^{\perp}$, this
implies that the only such ${\bf w}$ is the trivial one, i.e., ${\bf w}={\bf 0}$
and thus $\operatorname{Null}(M)=\{{\bf 0}\}$. By the Invertible
Matrix Theorem, it follows that $M$ is invertible. Hence, the convex
combination $f(a)M^{a}+f(b)M^{b}$ is generically invertible. 
\end{proof}

\subsection{Proof of Theorem \ref{prop:ident-3occasions}}\label{app:9}
\begin{proof}
As a first step, we will show that the Kruskal column rank\footnote{The Kruskal column rank of $M$ is the largest number $d$ such that
every set of $d$ columns of $M$ are independent.} of $M_{j}$ is at least $v_{j}$ generically, for each $j\in\left\{ 1,2,3\right\} $.
Take any $v_{j}$ types in $T$ and denote the set of these types
by $\bar{T}$. From \textit{(ii), }it follows that there is an injective
function $m:\bar{T}\to X_{j}$ such that $t\land_{j}m(t)$ for all
$t\in\bar{T}.$ Let $l(t)$ be such that $x_{l(t)}=m(t)$ and $k(t)$
be such that $t_{k(t)}=t$, for all $t\in\bar{T}$. Form a $v_{j}\times v_{j}$-submatrix
of $M_{j}$ using rows $\left\{ l(t):t\in\bar{T}\right\} $ and columns
$\left\{ k(t):t\in\bar{T}\right\} $. Observe that the determinant
of this matrix is a sum of monomials, one of which is $\prod_{t\in\bar{T}}M_{j}(l(t),k(t))$.
The set of parameter values contains points where, for all $t\in T$,
$M_{j}(l(t),k(t))=1$. At any of such points, the determinant of the
submatrix is equal to 1. Hence, the zero set of this determinant is
a proper subvariety and not the whole parameter space.

We have just shown that, for any selection of $v_{j}$ columns of
$M_{j}$, we can find a $v_{j}\times v_{j}$-minor of $M_{j}$ formed
of these columns such that its zero set is a proper subvariety. Hence,
for any $v_{j}$ columns of $M_{j}$, the columns are generically
independent. Denote by $\Theta_{j}$ the union of the zero sets of
these minors. Since this is a finite union, $\Theta_{j}$ is a proper
subvariety too. For any point in the parameter space that lies outside
of $\Theta_{j}$, all combinations of $v_{j}$ columns are independent.
This implies that the Kruskal column rank of $M_{j}$ is at least
$v_{j}$ generically.

For the second step, observe that outside of $\Theta_{0}=\cup_{j}\Theta_{j}$
the Kruskal column rank of $M_{j}$ is at least $v_{j}$ for all $j\in\left\{ 1,2,3\right\} .$
Denote by $\Theta$ the union of $\Theta_{0}$ and the set of parameter
values such that some component of $\pi$ is equal to 0. Clearly,
$\Theta$ is a proper subvariety. Outside $\Theta$, the sum of the
Kruskal column ranks of $M_{1}$, $M_{2}$, and $M_{3}$ is greater
than or equal to $v_{1}+v_{2}+v_{3}$ and all components of $\pi$
are strictly positive. Relying on the Kruskal's result \cite{Kruskal1977},
Corollary 2 of \cite{AllmanMatiasRhodes2009} says that, for $J=3$,
the parameters of the model \eqref{eq:choice-tensor-1-1} are uniquely
identifiable, up to label swapping, if the sum of the Kruskal column
ranks of $M_{1}$, $M_{2}$, and $M_{3}$ is greater than or equal
to $2r+2.$ From this and assumption \textit{(i)}, it follows that
the parameters outside $\Theta$ are uniquely identifiable, up to
label swapping. Hence, the parameters are generically identifiable,
up to label swapping. 
\end{proof}

\subsection{Proof of Theorem \ref{prop:ident-3occasions-type-state}}\label{app:10}
\begin{proof}
It follows from Theorem \ref{prop:gen-indep-cols-iff} that, for each
$j$, there is a proper subvariety $\Theta_{j}$ of the parameter
space such that outside $\Theta_{j}$ the Kruskal column rank of $M_{j}$
is $r$. Hence, outside of $\Theta_{0}=\cup_{j}\Theta_{j}$ the sum
of the Kruskal column ranks of the matrices $M_{1}$, $M_{2}$, and
$M_{3}$ is $3r$. Denote by $\Theta$ the union of $\Theta_{0}$
and the set of parameter values such that some component of $\pi$
is equal to 0. For any point outside $\Theta$ we can apply Corollary
2 of \cite{AllmanMatiasRhodes2009}. Since $\Theta$ is a proper subvariety,
the result follows. 
\end{proof}


\begin{thebibliography}{10}
\bibitem{AllmanMatiasRhodes2009}Allman, Elizabeth S., Catherine Matias,
and John A. Rhodes (2009) ``Identifiability of Parameters in Latent
Structure Models With Many Observed Variables,'' \emph{Annals of
Statistics} 37: 3099--3132.

\bibitem{Apesteguia2025}Apesteguia, José and Miguel-Anguel Ballester
(2025) \textquotedblleft A Measure of Behavioral Heterogeneity,\textquotedblright{}
\emph{American Economic Journal: Microeconomics}, forthcoming.

\bibitem{apeballu17}Apesteguia, José, Miguel-Anguel Ballester, and
Jay Lu (2017) \textquotedblleft Single-Crossing Random Utility Models,\textquotedblright{}
\emph{Econometrica} 85: 661--674.

\bibitem{Aumann1962}Aumann, Robert J. (1962) ``Utility Theory Without
the Completeness Axiom,'' \emph{Econometrica}: 445-462.

\bibitem{azrreh22}Azrieli, Yaron and John Rehbeck (2022) \textquotedblleft Marginal
Stochastic Choice,\textquotedblright{} \emph{arXiv}: 2208.08492 {[}econ.TH{]}.

\bibitem{barcoumoltei21}Barseghyan, Levon, Maura Coughlin, Francesca
Molinari, and Joshua C. Teitelbaum (2021) Heterogeneous Choice Sets
and Preferences, \emph{Econometrica }89(5): 2015-2048.

\bibitem{basmayqui68}Bassett, Lowell R., John S. Maybee, and James
Quirk (1968) ``Qualitative Economics and the Scope of the Correspondence
Principle,'' \emph{Econometrica} 36(3): 544--563. 

\bibitem{brusha95}Brualdi, Richard A. and Bryan L. Shader (1995)
``Matrices of Sign-Solvable Linear Systems,'' Cambridge University
Press, Cambridge.

\bibitem{Cappe2005}Cappé, Olivier, Eric Moulines, and Tobias Rydén
(2005) ``Inference in Hidden Markov Models,'' Springer, New York.

\bibitem{cartur24}Caradonna, Peter P. and Christopher Turansick (2026)
``Identification in Stochastic Choice,'' \emph{arXiv}: 2408.06547
{[}econ.TH{]}.

\bibitem{Chambers25}Chambers, Christopher P. and Christopher Turansick
(2025) \textquotedblleft The Limits of Identification in Discrete
Choice,\textquotedblright{} \emph{Games and Economic Behavior} 150:
537--551.

\bibitem{Cox97}Cox, David A., John Little, and Donal O'Shea (1997)
\textquotedblleft Ideals, Varieties, and Algorithms,\textquotedblright{}
2nd ed, Springer, New York.

\bibitem{DardanoniManziniMariottiTyson2020}Dardanoni, Valentino,
Paola Manzini, Marco Mariotti, and Christopher J. Tyson (2020) ``Inferring
Cognitive Heterogeneity from Aggregate Choices,'' \emph{Econometrica}
88(3): 1269-1296.

\bibitem{darmanmarpettys24}Dardanoni, Valentino, Paola Manzini, Marco
Mariotti, Henrik Petri, and Christopher J. Tyson (2024) ``Mixture
Choice Data: Revealing Preferences and Cognition,'' \emph{Journal
of Political Economy} 131(3): 687-715.

\bibitem{doisai23}Doignon, Jean-Paul and Kota Saito (2023) \textquotedblleft Adjacencies
on Random Ordering Polytopes and Flow Polytopes,\textquotedblright{}
\emph{Journal of Mathematical Psychology}, 114: 102768.

\bibitem{Dubra2004}Dubra, Juan, Fabio Maccheroni, and Efe A. Ok.
(2004) ``Expected Utility Theory Without the Completeness Axiom,''
\textit{Journal of Economic Theory} 115.1: 118-133.

\bibitem{Ephraim2002}Ephraim, Yariv and Neri Merhav (2002) \textquotedblleft Hidden
Markov Processes,\textquotedblright{} \emph{IEEE Transactions on Information
Theory} 48, 1518-1569.

\bibitem{filmas23}Filiz-Ozbay, Emel and Yusufcan Masatlioglu (2023)
\textquotedblleft Progressive Random Choice,\textquotedblright{} \emph{Journal
of Political Economy} 131: 716--750.

\bibitem{fish98}Fishburn, Peter C. (1998) \textquotedblleft Stochastic
Utility,\textquotedblright{} in \emph{Handbook of Utility Theory},
ed. by S. Barbera, P. Hammond, and C. Seidl, Kluwer Dordrecht; 273--318.

\bibitem{gor64}Gorman, William M. (1968) ``More Scope for Qualitative
Economics,'' \emph{The Review of Economic Studies} 31: 65-68. 

\bibitem{Hall1935}Hall, Philip (1935) ``On Representatives of Subsets,''
\emph{Journal of the London Mathematical Society }1: 26--30.

\bibitem{Koopmans1950a}Koopmans, Tjalling C. and Olav Reisersøl (1950)
\textquotedblleft The identification of Structural Characteristics,\textquotedblright{}
\emph{The Annals of Mathematical Statistics} 21: 165--181.

\bibitem{Koopmans1950b}Koopmans, Tjalling C. (1950) \textquotedblleft Statistical
Inference in Dynamic Economic Models,\textquotedblright{} Wiley, New
York.

\bibitem{KopsManziniMariottiPasichnichenko2022}Kops, Christopher
J., Paola Manzini, Marco Mariotti, and Illia Pasichnichenko (2025)
``Revealing Features from Pareto Optimal Choice,'' Working paper.

\bibitem{Kruskal1977}Kruskal, Joseph B. (1977) ``Three-Way Arrays:
Rank and Uniqueness of Trilinear Decompositions, With Application
to Arithmetic Complexity and Statistics,'' \emph{Linear Algebra and
its Applications} 18: 95--138.

\bibitem{Lancaster62}Lancaster, Kelvin J. (1962) ``The Scope of
Qualitative Economics,'' \emph{The Review of Economic Studies} 29(2):
99--123.

\bibitem{Lancaster1966}Lancaster, Kelvin J. (1966) ``A New Approach
to Consumer Theory,'' \emph{Journal of Political Economy} 74(2):
132-157.

\bibitem{Lindsay1995}Lindsay, Bruce G. (1995) ``Mixture Models:
Theory Geometry and Applications,'' \emph{NSF-CBMS Regional Conference
Series in Probability and Statistics} 5, IMS, Hayward, CA.

\bibitem{Lu2019}Lu, Jay (2019) ``Bayesian Identification: A Theory
for State-Dependent Utilities,'' \emph{American Economic Review}
109(9): 3192-3228.

\bibitem{manmar12}Manzini, Paola and Marco Mariotti (2012) ``Categorize
Then Choose: Boundedly Rational Choice and Welfare,'' \textit{Journal
of the European Economic Association} 10.5: 1141-1165.

\bibitem{manmar18}Manzini, Paola and Marco Mariotti (2018) \textquotedblleft Dual
Random Utility Maximisation,\textquotedblright{} \emph{Journal of
Economic Theory} 177: 162--182.

\bibitem{manmarpet19}Manzini, Paola, Marco Mariotti, and Henrik Petri
(2019) ``Corrigendum to \textquotedblleft Dual Random Utility Maximisation\textquotedblright{}
{[}J. Econ. Theory 177 (2018) 162--182{]},'' \emph{Journal of Economic
Theory} 184: 104944.

\bibitem{Masatlioglu2012}Masatlioglu, Yusufcan, Daisuke Nakajima,
and Erkut Y. Ozbay (2012) ``Revealed Attention,'' \textit{American
Economic Review} 102.5: 2183-2205.

\bibitem{McLachlan}McLachlan, Geoffrey J. and David Peel (2000) ``Finite
Mixture Models,'' Wiley, New York.

\bibitem{Rhodes2010}Rhodes, John A. (2010) \textquotedblleft A Concise
Proof of Kruskal's Theorem on Tensor Decomposition,\textquotedblright{}
\textit{Linear Algebra and its Applications} 432: 1818-1824.

\bibitem{sam47}Samuelson, Paul A. (1947) ``Foundations of Economic
Analysis,'' Harvard University Press, Cambridge, Massachusetts

\bibitem{tur22}Turansick, Christopher (2022) \textquotedblleft Identification
in the Random Utility Model,\textquotedblright{} \emph{Journal of
Economic Theory} 203,105489.

\bibitem{yil22}Yildiz, Kemal (2023) \textquotedblleft Foundations
of Self-Progressive Choice Models,\textquotedblright{} \emph{arXiv}:
2212.13449 {[}econ.TH{]}.

\end{thebibliography}
\end{document}